\documentclass{article}

\usepackage[preprint]{tmlr}

\usepackage{amsmath}
\usepackage{amssymb}
\usepackage{amsthm}
\usepackage{graphicx}
\usepackage{booktabs}
\usepackage{array}
\usepackage{adjustbox}

\usepackage{pifont}
\usepackage{tikz}
\usetikzlibrary{positioning,arrows.meta,calc,backgrounds,fit,shapes.geometric}
\definecolor{ink}{HTML}{1F2937}
\definecolor{inFill}{HTML}{EEF2F7}\definecolor{inLine}{HTML}{94A3B8}
\definecolor{pdFill}{HTML}{DBEAFE}\definecolor{pdLine}{HTML}{2563EB}
\definecolor{rsFill}{HTML}{FEF3C7}\definecolor{rsLine}{HTML}{D97706}
\definecolor{tsFill}{HTML}{DCFCE7}\definecolor{tsLine}{HTML}{16A34A}
\definecolor{coFill}{HTML}{EDE9FE}\definecolor{coLine}{HTML}{7C3AED}
\definecolor{fbFill}{HTML}{F1F5F9}\definecolor{fbLine}{HTML}{64748B}
\definecolor{ouFill}{HTML}{E2E8F0}\definecolor{ouLine}{HTML}{334155}
\definecolor{noGood}{HTML}{DC2626}
\definecolor{lab}{HTML}{6B7280}

\usepackage{hyperref}
\hypersetup{colorlinks=true, linkcolor=blue, citecolor=blue, urlcolor=blue,
  pdftitle={NeuroForge: Self-Auditing Neural CFD Surrogates with Calibrated Physics-Residual Trust},
  pdfauthor={Ali Jabbary, Kasra Ghanavati}}

\theoremstyle{plain}
\newtheorem{theorem}{Theorem}
\newtheorem{proposition}{Proposition}

\title{NeuroForge: Self-Auditing Neural CFD Surrogates\\ with Calibrated Physics-Residual Trust}

\author{Ali Jabbary \\
  \addr Department of Mechanical Engineering, Urmia University \\
  \email st\_a.jabbary@urmia.ac.ir
  \AND
  Kasra Ghanavati \\
  \addr School of Computing and Mathematical Sciences, University of Greenwich \\
  \email kg1111r@gre.ac.uk}

\begin{document}

\maketitle

\begin{abstract}
Machine-learning surrogates for computational fluid dynamics (CFD) predict steady
flow fields far faster than classical solvers, but emit a
single field with no built-in way to know whether to trust it---especially out of
distribution. We make the surrogate \emph{audit itself} against the governing physics: we
compute the discretised steady-state Reynolds-averaged (RANS) residual of the prediction
and ask what jobs it can do. Our central finding is a clean \textbf{two-way dissociation}: the physics
residual is a reliable, backbone-robust \textbf{trust signal} (it tells you \emph{where}
the prediction is wrong) but a poor \textbf{correction objective} (it does not tell you
\emph{how} to fix it). \emph{As a trust signal}, the residual's per-case rank correlation
with field error is consistently positive across three architecturally distinct backbones
and a second dataset of laminar bluff bodies; it flags the worst-decile cases with AUROC
$\approx 0.9$, and a distribution-free split-conformal layer attains its target coverage. \emph{As a correction objective}, the
residual fails: reducing it does not reduce field error. The monotone-residual acceptance
gate built on it, by contrast, does help. Separately, a learned deep-equilibrium corrector trained
toward ground truth lowers volume-field error on a state-of-the-art backbone; a controlled
ablation that removes the residual input matches this gain, so the improvement comes from
the learned correction, not from conditioning on the residual. The contribution is a
self-auditing, calibrated trust layer, the residual's two roles (trust signal yes,
correction objective no), and the learned correction it accompanies, demonstrated on a
competitive surrogate and released as the open-source \texttt{neuroforge-cfd} package.

\vspace{1ex}
\end{abstract}

\section{Introduction}

Classical CFD resolves the governing PDEs by iterating a discretised system until its
residuals fall below a tolerance. It is accurate and trusted, but slow: meshing and a
converged steady-RANS solve over a new geometry can take minutes to hours, throttling
the rapid design-space exploration that early engineering most needs. Neural
surrogates promise the opposite trade---millisecond inference---and a vigorous
literature now maps geometry plus boundary conditions directly to flow fields with
neural operators and attention models.

The problem is \textbf{trust}. A trained surrogate is a one-shot function: it emits a
field with no self-assessment, no convergence signal, and no recovery path when it
extrapolates. On a geometry far from the training set---the common case in design---the
user cannot distinguish a faithful prediction from a confident hallucination.

The governing physics offers an obvious lever. The steady incompressible RANS residual
of any candidate field is computable directly from the field, with no ground truth.
A neural-CFD pipeline could use this residual in three distinct ways: as a
\emph{trust signal} that flags \emph{which} predictions to distrust (a case-level error
proxy that drives a conformal trust layer), as a \emph{correction objective} that a step
should minimise, and as an \emph{acceptance gate} that a step must not worsen. These are
not the same role, and---this is the crux---they do not earn the same verdict. Put plainly, the residual tells you
\emph{where} a prediction is wrong, not \emph{how} to fix it. This paper separates the two
roles and answers each empirically on a SOTA backbone, and---separately---deploys a
learned, supervised corrector whose gain we attribute by ablation to the learned
correction, not to the residual.

\paragraph{Research question.} \emph{In which of its roles---trust signal, correction
objective, and acceptance gate---is the steady-RANS physics residual of a neural-CFD
prediction useful, and does the answer survive on a competitive backbone?}

\paragraph{Thesis (one sentence).} \textbf{\emph{The steady-RANS physics residual is a
reliable, conformally-calibrated, backbone-robust case-level trust signal---it tells you
where a prediction is wrong---even though minimising the residual directly remains a poor
correction objective; the learned correction we deploy reduces error on a SOTA
surrogate, but a controlled ablation attributes that gain to the learned correction, not to
feeding it the residual.}}

\paragraph{Self-auditing, defined.} We call a surrogate \emph{self-auditing} if, from its
own prediction and inputs alone---no ground truth, no reference solver---it supplies three
measurable things: \textbf{(i)} a case-level \emph{trust score} whose ranking of
predictions by true error is validated; \textbf{(ii)} a \emph{calibrated error band} with
a distribution-free guarantee; and \textbf{(iii)} an \emph{accept/reject decision} whose
efficiency is measured against an oracle. The title is a claim, and the paper is its
proof obligation: we instantiate all three components with the physics residual at the
centre---(i) positive residual--error rank correlation across three backbones and two
datasets (\autoref{sec:v2}); (ii) split-conformal coverage at target, stable across 20
calibration re-draws (\autoref{sec:conformal}); (iii) near-oracle triage, recovering
$\approx 91\%$ of the achievable error reduction at a $10\%$ rejection budget
(\autoref{fig:risk_coverage}; the deciding score fuses the residual with the surrogate's
own ensemble spread---neither needs ground truth---and the residual alone still recovers
$\approx 67\%$)---and then quantify the audit's limits: the residual-floor
theorem (\autoref{sec:residual_floor}) characterises exactly what the monitored residual
cannot see, and the failed \emph{objective} role marks what the audit does \emph{not}
license---the audit tells you where to look, not how to fix.

We arrive at this thesis the hard way. We built the full
predict$\rightarrow$verify$\rightarrow$estimate-uncertainty$\rightarrow$correct$\rightarrow$fall-back
pipeline, pre-registered hypotheses, and ran a 3-seed ablation on real AirfRANS plus an
out-of-distribution study, two formal certificates, and---in a second phase---the same
trust layer and corrector on a SOTA Transolver backbone. The trust signal and the
supervised corrector survive on the strong backbone; residual \emph{minimisation} is the
part that disappoints---the iteration sweep drives residual and error apart---and we report
that as a finding rather than hide it. The acceptance gate built on the same residual is a
\emph{separate} mechanism, and measured directly it works (\autoref{sec:iters}). Two honest findings accompany the corrector:
\emph{correction quality is backbone-dependent}---the supervised corrector that helps on
the SOTA backbone is flat on the weak grid backbone---and a controlled ablation shows the
gain is \emph{not} sourced from the residual input (a null corrector with the residual
channels zeroed matches or slightly exceeds it), so we attribute the gain to the learned
correction, not to residual-conditioning. We report both.

\subsection{Contributions}

In priority order:

\begin{enumerate}
\item \textbf{A backbone-agnostic trust layer over a consistently positive residual detector (headline).} The physics residual is
a calibrated, \emph{case-level} error detector across three architecturally-distinct backbones
tested. On the SOTA
Transolver backbone its per-case rank (Spearman) correlation with field error is
$0.611{\pm}0.054$ over five seeds (\autoref{sec:v2}, \autoref{tab:v2}); on a much weaker grid backbone it is
$\approx 0.40$ bare and a learned corrector lifts it to $0.83$ (\autoref{sec:indist-ablation});
and on a param-matched message-passing GNN (MeshGraphNet, bare) it is $0.851{\pm}0.058$
(\texttt{results/mgn/mgn\_results.json}, a high value that partly reflects that model's wide
error spread, which is easy to rank),
so the detector is consistently \emph{positive} across all three, though its \emph{strength}
varies ($\rho$ from $0.40$ to $0.85$; spatial localisation is weaker and
region-scale---patch pooling doubles it to $\rho{\approx}0.32$,
\autoref{sec:method}). Building on it---and depending only on a model's predictions and
$\sigma$, hence backbone-agnostic---a split-conformal trust layer attains
target coverage on the SOTA backbone ($0.902{\pm}0.008$ at the $0.90$ target,
\autoref{sec:v2}); with a deep-ensemble $\sigma$ the band is input-adaptive on that same
near-deterministic model (ECE $0.074$, \autoref{tab:uq}), and on the weak backbone it
holds under distribution shift, including a regime where the bare model's signal collapses
(\autoref{sec:ood}, \autoref{sec:conformal}).
The signal also supports \emph{selective prediction}: it detects worst-decile-error cases
with AUROC $0.87$--$0.91$ on the AirfRANS arms ($0.83$--$0.98$ including the DeepCFD OOD
seeds), and rank-fused with the ensemble $\sigma$ reaches AUROC $0.905$, recovering
$\approx 91\%$ of the oracle's achievable error reduction at a $10\%$ rejection budget
(\autoref{fig:risk_coverage}).
The layer depends only on a model's predictions and uncertainty, so it is backbone-robust.

\item \textbf{Learned correction works on a strong backbone.} A learned DEQ corrector trained
(supervised) toward ground truth reduces volume-field MSE on the SOTA backbone on
\textbf{all five seeds}: $\texttt{mse\_u}$ $0.133\rightarrow0.122$ ($-8\%$),
$\texttt{mse\_v}$ $0.096\rightarrow0.076$ ($-21\%$), $\texttt{mse\_p}$ $644\rightarrow485$
($-25\%$, gate-verified), with internal force-ranking self-consistency held
($\rho_{C_l}=0.999$; against \emph{official} labels drag ranking is recovered at
$\rho_D=0.84$, and a repaired control-volume integrator recovers official \emph{lift}
nearly exactly from the predicted fields, $\rho_L=0.998{\pm}0.001$ with per-seed median
magnitude errors of $3.6$--$3.9\%$, \autoref{sec:v2}). The backbone is a published SOTA architecture (Transolver),
so the corrector improves a competitive surrogate.\footnote{Throughout, ``SOTA'' refers
to Transolver's status as a published state-of-the-art architecture at the time of our
experiments (ICML 2024); an independent 2026 benchmark still highlights Transolver(++)
among the strongest aerodynamic surrogates \citep{scherz2026evaluation}, while newer
architectures have since advanced the AirfRANS accuracy frontier \citep{rigotti2026gist}.
Our claims concern the trust layer and corrector demonstrated \emph{on} this competitive
backbone, not an accuracy-leaderboard position.} The deployed corrector
ingests the residual as an input channel, but a controlled ablation that zeros that input at
train and inference (5 seeds) matches or slightly exceeds the deployed corrector, so we
attribute the gain to the learned correction, not to residual-conditioning (\autoref{sec:v2},
\texttt{results/control/w1\_capture.json}).

\item \textbf{Residual-as-objective fails---but the gate built on it works (a separated
verdict).} Used as a correction \emph{objective}, the residual
does not work. Sweeping the correction iterations \emph{raises} the PDE residual
($0.11\rightarrow0.62$) while \emph{lowering} field error
($\texttt{mse\_u}$ $3.92\rightarrow2.29$): the two move in opposite directions
(\autoref{sec:iters}, \autoref{fig:fig4}). The \emph{acceptance gate} is a distinct
mechanism and earns a distinct verdict: measured over $1200$ gated steps it admits a step on
$99.8\%$ of deployed cases (the full correction on half of them), and the admitted step
lowers true field error on $89.3\%$ of them by a median $5.8\%$ (\autoref{sec:iters},
\texttt{results/control/acceptance\_gate.json}). The monotone-residual guarantee is thus
cheap, provably non-worsening \emph{and} useful---it is residual \emph{minimisation}, not
the gate, that fails. Minimising the residual is not minimising the error---which is
exactly why the deployed corrector supervises toward truth rather than minimising the
residual. The residual-floor theorem locates this failure as a \emph{regime boundary}
rather than a contradiction of prior successful residual-driven correctors: correction
works where the discrete residual is (near-)exact at the truth
\citep{huang2025physicscorrect, learned2023residualcorrection}, and fails where the
monitored residual carries a floor, as here (\autoref{sec:residual_floor}).

\item \textbf{Backbone-dependence, reported as a finding (not hidden).} The same corrector
that delivers $-9$ to $-25\%$ on the SOTA backbone is flat-to-unhelpful on the weak grid
backbone (\autoref{sec:indist-ablation}, \autoref{tab:indist}), where it trades volume
accuracy for surface fidelity and a stronger trust signal. Correction quality is therefore
backbone-dependent; we report both regimes side by side and the SOTA result is the one that
supports the learned-correction claim. For honest context we also report a matched-budget
baseline study (\autoref{sec:baseline}, \autoref{tab:transolver}): the grid backbone trails
Transolver by $4$--$60\times$, while the trust layer and corrector are demonstrated
\emph{on} the competitive Transolver backbone in \autoref{sec:v2}.

\item \textbf{A reproducible, CPU-first open-source package} (\texttt{neuroforge-cfd})
with a frozen I/O contract, interchangeable backbones, a zero-download synthetic
potential-flow data generator, an AirfRANS loader, the ablation/certificate harness,
and the figures---so every claim runs on a laptop and every certificate is reproducible
from a committed script.
\end{enumerate}

We are explicit throughout about what is \textbf{measured} versus \textbf{assumed},
retain the retraction of stale single-seed numbers from earlier drafts
(\autoref{sec:indist-ablation}), and make no out-of-distribution coverage \emph{guarantee};
we make a competitiveness claim only for the Transolver backbone our trust layer runs on
(\autoref{sec:v2}), not for the grid backbone.

\paragraph{How to read this paper.} The argument is long because the controls are the
argument, so it is worth saying where each kind of evidence lives.
\autoref{sec:method} defines the audit and proves its two properties---contraction under
the acceptance gate, and split-conformal coverage---and \autoref{sec:residual_floor}
states the residual-floor theorem that bounds what the audit can detect at all.
\emph{Does it work?} \autoref{sec:v2} (trust signal and learned correction on a
competitive backbone) and \autoref{sec:conformal} (calibrated bands, selective
prediction). \emph{Does it keep working?} \autoref{sec:v2} and \autoref{sec:ood} across
three backbone families, two datasets and a regime shift, with five-seed spreads and
case-level bootstrap intervals throughout, and \autoref{sec:indist-ablation} for the
controls---including the ones that falsified claims we had previously made.
\emph{What does it cost?} \autoref{sec:cost} times every stage of the deployed pipeline.
\emph{Can it be checked?} Every number here maps to a committed script and result file
through \texttt{docs/REPRODUCE.md}, with a manifest recording seeds, environment and
SHA-256 hashes. Readers who want only the central result can read
\autoref{sec:v2}, \autoref{sec:iters} and \autoref{sec:cost}.

\section{Related Work}

\paragraph{Neural operators (FNO / Geo-FNO).} Fourier Neural Operators learn a
resolution-agnostic mapping between function spaces by parameterising a global
convolution in the spectral domain \citep{li2021fno}. Geo-FNO \citep{li2022geofno}
extends this to irregular geometries via a learned deformation to a latent uniform grid
where the FFT is valid. These are strong \emph{one-shot} operators with no
self-assessment; NeuroForge implements both as backbones and adds the trust layer
around them.

\paragraph{Point-cloud and physics-attention operators (DoMINO, Transolver).} DoMINO
\citep{ranade2025domino} is a decomposable multi-scale point-cloud operator for external
aerodynamics on large industrial meshes. Transolver \citep{wu2024transolver} replaces
quadratic attention with attention over learnable \emph{physics slices}, giving linear
cost and strong accuracy on PDE benchmarks; Transolver++ \citep{luo2025transolverpp}
scales it to million-scale meshes. We use Transolver in two ways: as a matched-budget
baseline for our grid backbone (\autoref{sec:baseline}), and as the SOTA backbone our
trust layer and corrector actually run on (\autoref{sec:v2}), where the residual is a
calibrated detector and a supervised learned corrector reduces field MSE.

\paragraph{Learned solver-correctors and the ``fixer'' premise.} Several lines of work
use the discretised residual or a coarse-solution error as a \emph{correction
objective}: learned PDE solvers with convergence guarantees \citep{hsieh2019neuralpde},
deep-equilibrium operators for steady PDEs \citep{marwah2023fnodeq}, iterative refiners
\citep{lippe2023pderefiner}, learned residual-error correctors
\citep{learned2023residualcorrection}, and---most recently---training-free correction
that solves a linearised inverse problem on the PDE residual at each rollout step,
reporting up to $100\times$ error reduction on transient benchmarks
\citep{huang2025physicscorrect}. PINNs \citep{raissi2019pinn} embed the residual
in the training loss. NeuroForge implements a contractive DEQ corrector and a
feed-forward residual-conditioned corrector squarely in this family. \textbf{Our
contribution to this line is to test the residual-as-correction premise head-on and find it
does not hold.} As a correction \emph{objective}---where a step must reduce the
residual---residual minimisation does not track error minimisation on a real RANS benchmark
(the residual rises as error falls). And when we
feed the residual as an input \emph{feature} to a corrector supervised toward ground truth, a
controlled ablation shows the residual input is not the source of the gain: zeroing it at
train and inference (3 seeds) matches or slightly exceeds the deployed corrector
(\autoref{sec:v2}, \texttt{results/control/w1\_capture.json}). The residual's value is as a
\emph{trust signal}---it tells you \emph{where} a prediction is wrong---not as a correction
mechanism telling you \emph{how} to fix it; the learned correction does the fixing.
The prior successes and our negative are consistent, and the residual-floor theorem
(\autoref{sec:residual_floor}) marks the boundary between them: residual-driven correction
works where the discrete residual is (near-)exact at the true solution---resolved
transient PDEs stepped on their training grid \citep{huang2025physicscorrect}, variational
problems with certified residuals \citep{learned2023residualcorrection}---and fails where
the monitored residual carries a nonzero floor at the dataset truth, as for our
under-resolved steady-RANS closure monitor. Our negative is scoped to the latter,
deployment-realistic regime.

\paragraph{Residual-based, data-free error signals.} Closest to our \emph{trust-signal}
half is recent work that reads the PDE residual as a reference-free indicator of prediction
error. \citet{roy2025anchor} use a physics-informed, residual-based error estimator to detect
and correct accumulating error in neural-operator \emph{time-marching}, triggering a classical
solver when the estimator fires; \citet{song2026structureaware} design an epistemic-UQ scheme
whose uncertainty bands are aligned with localised residual structure in operator surrogates.
\citet{gopakumar2025pre} go furthest in the trust direction: they use the PDE residual itself
as a conformal nonconformity score, obtaining data-free coverage guarantees---in
\emph{residual} space, with the transfer to solution space left open as a set-propagation
problem. Our residual-floor theorem (\autoref{sec:residual_floor}) supplies a concrete
mechanism for that gap in the discrete steady-RANS setting: the monitored residual has a
nonzero floor at the reference solution and a kernel of undetectable modes, so residual-space
smallness cannot by itself certify solution-space accuracy. The question we put to the residual---can a quantity computable without the true
solution certify that solution's accuracy?---is the defining question of
a-posteriori error estimation, and the classical answer is instructive: residual
magnitude is not itself an error estimate, and becomes one only when weighted by the
solution of an adjoint problem measuring how strongly each local residual propagates
into the functional of interest \citep{beckerrannacher2001dwr}. The neural analogue
keeps that structure: \citet{hillebrecht2025posteriori} bound the true error of a
PDE-defined PINN by its residual scaled by a stability constant of the underlying
operator. Both results are conditional on a property our monitor does not have---that
the residual vanishes at the exact solution---and the residual-floor theorem is
precisely a statement of how that condition fails for an under-resolved discrete
steady-RANS operator. That is why we use the residual to \emph{rank} cases rather than
to \emph{bound} them, and obtain the bound from conformal calibration instead.
From the theory side,
\citet{mukherjee2026certification} prove that vanishing residuals guarantee solution
convergence only under additional compactness assumptions; our setting is the practical
complement, where the discrete monitored residual provably does \emph{not} vanish at the
dataset truth. We share the premise that the residual is a data-free error signal and do not
claim to originate it. Our distinct contributions are (i)~the head-to-head \emph{dissociation}---the residual is a
good trust signal but a poor correction objective, established on the same steady-state RANS
benchmark; (ii)~the \emph{residual-floor theorem} quantifying the operator-specific detection
limit and the undetectable kernel modes; and (iii)~a split-conformal certificate calibrated on
the \emph{deployed, corrected} field. Our setting is steady-state external aerodynamics with
per-case spatial trust, rather than time-marching error control.

\paragraph{Uncertainty quantification and conformal prediction for PDE surrogates.}
Deep ensembles \citep{lakshminarayanan2017ensembles} and MC-dropout
\citep{gal2016dropout} are standard epistemic estimators; conformal prediction gives
distribution-free coverage and has been applied to operator learning
\citep{ma2024uqno}. Within this journal the same combination is being pursued from both
the PINN and the operator side: \citet{yu2026conformalpinn} conformalise the output of
a physics-informed network, and \citet{garg2025dfuq} give distribution-free bands for
neural operators. Both calibrate the \emph{raw} surrogate; we calibrate the field the
user actually receives---the corrected one---which matters here because the correction
step changes the very error distribution the quantile is fitted to
(\autoref{sec:conformal}). Concurrently with this work, \citet{jia2026multigranularity} develop
multi-granularity conformal calibration (case-level quantile regression plus
residual-normalised point-level bands) for neural-operator automotive surrogates on
DrivAerML; their nonconformity scales are purely data-driven, whereas our trust signal is
the physics residual itself, and we additionally test---and reject---its use as a
correction objective. We do \textbf{not} propose a new UQ method. Our trust layer applies
split-conformal calibration to MC-dropout uncertainty; the new content is the
\emph{characterisation} of when this holds for CFD surrogates---in particular that
coverage survives out-of-distribution empirically via uncertainty inflation, while the
formal guarantee does not transfer once exchangeability fails.

\paragraph{Benchmarks (AirfRANS).} AirfRANS \citep{bonnet2022airfrans} provides
$\sim 1000$ incompressible steady-RANS simulations over NACA airfoils with
\texttt{full}, \texttt{scarce}, \texttt{reynolds}, and \texttt{aoa} splits designed to
probe generalisation. We use \texttt{full} for in-distribution evaluation and the
regime-disjoint \texttt{reynolds}/\texttt{aoa} splits for the out-of-distribution study.

\subsection{Positioning}

The closest prior art treats the physics residual primarily as a \emph{correction
objective} (PINNs, FNO-DEQ, learned correctors) or treats UQ in isolation from physics.
NeuroForge's distinctive question is which of two roles the residual can play, evaluated
head-to-head on the same model and data and then on a SOTA backbone: as a trust
\emph{signal} (detect error, drive conformal calibration) and as a correction
\emph{objective}, with the acceptance \emph{gate} built on it judged separately. Our
answer---signal yes (it tells you \emph{where} the prediction is wrong), objective no, gate
yes---and the resulting calibrated trust layer are the contribution.
Separately, a learned, supervised corrector reduces field error on the strong backbone, but a
controlled ablation attributes that gain to the learned correction rather than to the residual
input. The framework's no-harm/contraction machinery (\autoref{sec:method})
is retained because it is correct and reproducible; the acceptance-gate guarantee in
particular is real and, as we now measure, materially useful---it admits a step on $99.8\%$
of deployed cases and that step lowers field error on $89.3\%$ of them
(\autoref{sec:iters}).
\autoref{tab:positioning} summarises the positioning: each neighbouring work holds one
or two of the audit's components, and none combines a validated physics-residual trust
signal with distribution-free calibration, measured selective prediction, a head-on test
of the residual's objective role, and a theory of the signal's limits.

\begin{table}[t]
\centering
\small
\caption{Positioning against the closest 2024--2026 work. Columns: \emph{Trust} =
physics-residual trust signal validated against true error; \emph{Calib.} =
distribution-free calibrated bands; \emph{Select.} = selective prediction measured
(risk--coverage / AUROC); \emph{Object.} = residual-as-correction-objective tested;
\emph{Theory} = formal theory of the signal's limits; \emph{$\ge$3 bb.} = three or more
backbone families; \emph{$\ge$2 data} = two or more datasets/regimes.
$\bullet$ = demonstrated, $\circ$ = partial, -- = not a goal of / not reported by that
work. In the \emph{Object.} column, PhysicsCorrect and the variational corrector
\emph{succeed} in regimes where the discrete residual is (near-)exact at the truth; we
test the floored steady-RANS regime where it \emph{fails}---the residual-floor theorem
(\autoref{sec:residual_floor}) is the boundary between the two.}
\label{tab:positioning}
\begin{adjustbox}{max width=\textwidth}
\begin{tabular}{l c c c c c c c}
\toprule
 & Trust & Calib. & Select. & Object. & Theory & $\ge$3 bb. & $\ge$2 data \\
\midrule
PRE conformal \citep{gopakumar2025pre} & $\bullet$ & $\bullet$ & -- & -- & -- & $\bullet$ & $\bullet$ \\
Structure-aware UQ \citep{song2026structureaware} & $\circ$ & -- & -- & -- & -- & -- & -- \\
ANCHOR \citep{roy2025anchor} & $\bullet$ & -- & $\circ$ & -- & -- & -- & -- \\
PhysicsCorrect \citep{huang2025physicscorrect} & -- & -- & -- & $\bullet$ & $\circ$ & $\bullet$ & $\bullet$ \\
Variational corrector \citep{learned2023residualcorrection} & -- & -- & -- & $\bullet$ & $\circ$ & -- & -- \\
Conformal DrivAerML \citep{jia2026multigranularity} & -- & $\bullet$ & -- & -- & -- & $\circ$ & -- \\
Error certification \citep{mukherjee2026certification} & -- & -- & -- & -- & $\bullet$ & -- & -- \\
\midrule
\textbf{This work} & $\bullet$ & $\bullet$ & $\bullet$ & $\bullet$ & $\bullet$ & $\bullet$ & $\bullet$ \\
\bottomrule
\end{tabular}
\end{adjustbox}
\end{table}

\section{Method}
\label{sec:method}

\subsection{Pipeline overview}

The NeuroForge framework wraps a one-shot backbone in a verify--estimate--correct--fall-back loop
(\autoref{fig:pipeline}).
Throughout, the physics residual plays \textbf{two distinct roles with two distinct
verdicts} (\autoref{sec:experiments}): (i)~as a \emph{trust signal} it detects error and
drives the conformal layer (reliable, backbone-robust---it tells you \emph{where} the
prediction is wrong); (ii)~as a correction \emph{objective} it fails---though the
\emph{acceptance gate} built on it does work (\autoref{sec:iters}). Separately, a learned, supervised corrector improves a SOTA backbone
(\autoref{sec:v2}); the deployed corrector ingests the residual as an input channel, but a
controlled ablation attributes its gain to the learned correction, not to the residual input.
We describe all components for completeness and flag, at each one, which role it instantiates.

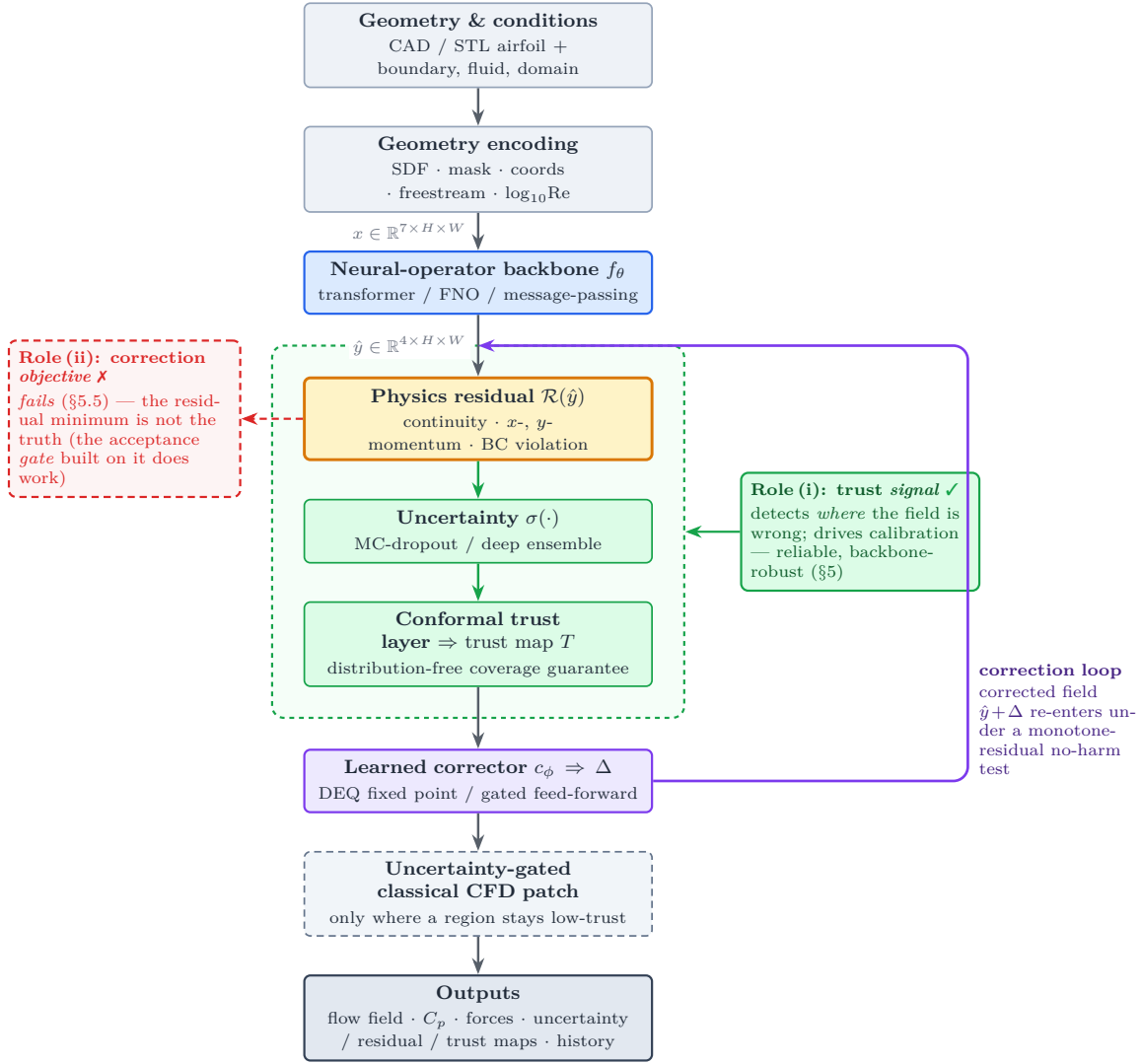
\begin{figure}[t]
\centering
\resizebox{0.92\linewidth}{!}{%
\begin{tikzpicture}[
  >={Stealth[length=2.4mm]},
  font=\footnotesize, line width=0.7pt, node distance=5.5mm and 12mm,
  base/.style={rounded corners=2.5pt, draw, align=center, inner sep=4pt,
               minimum height=9mm, text width=48mm, text=ink},
  stage/.style   ={base, fill=inFill, draw=inLine},
  predict/.style ={base, fill=pdFill, draw=pdLine, line width=0.9pt},
  pivot/.style   ={base, fill=rsFill, draw=rsLine, line width=1.2pt},
  trust/.style   ={base, fill=tsFill, draw=tsLine},
  correct/.style ={base, fill=coFill, draw=coLine, line width=0.9pt},
  fallback/.style={base, fill=fbFill, draw=fbLine, dash pattern=on 3pt off 2pt},
  output/.style  ={base, fill=ouFill, draw=ouLine, line width=0.9pt},
  flow/.style    ={->, line width=1.0pt, draw=ink!75},
  elab/.style    ={font=\scriptsize, text=lab, inner sep=1.5pt, fill=white},
  roletag/.style ={rounded corners=3pt, font=\scriptsize, inner sep=4pt, align=left,
                   line width=0.9pt},
]
\node[stage]   (geo) {\textbf{Geometry \& conditions}\\[1pt]{\scriptsize CAD / STL airfoil $+$ boundary, fluid, domain}};
\node[stage, below=of geo] (enc) {\textbf{Geometry encoding}\\[1pt]{\scriptsize SDF $\cdot$ mask $\cdot$ coords $\cdot$ freestream $\cdot$ $\log_{10}\!\mathrm{Re}$}};
\node[predict, below=of enc] (bb) {\textbf{Neural-operator backbone} $f_\theta$\\[1pt]{\scriptsize transformer / FNO / message-passing}};
\node[pivot, below=9mm of bb] (res) {\textbf{Physics residual} $\mathcal{R}(\hat{y})$\\[1pt]{\scriptsize continuity $\cdot$ $x$-, $y$-momentum $\cdot$ BC violation}};
\node[trust, below=of res] (unc) {\textbf{Uncertainty} $\sigma(\cdot)$\\[1pt]{\scriptsize MC-dropout / deep ensemble}};
\node[trust, below=of unc] (conf) {\textbf{Conformal trust layer} $\Rightarrow$ trust map $T$\\[1pt]{\scriptsize distribution-free coverage guarantee}};
\node[correct, below=9mm of conf] (corr) {\textbf{Learned corrector} $c_\phi \Rightarrow \Delta$\\[1pt]{\scriptsize DEQ fixed point / gated feed-forward}};
\node[fallback, below=of corr] (fb) {\textbf{Uncertainty-gated classical CFD patch}\\[1pt]{\scriptsize only where a region stays low-trust}};
\node[output, below=of fb] (out) {\textbf{Outputs}\\[1pt]{\scriptsize flow field $\cdot$ $C_p$ $\cdot$ forces $\cdot$ uncertainty / residual / trust maps $\cdot$ history}};
\draw[flow] (geo)  -- (enc);
\draw[flow] (enc)  -- (bb)  node[elab, midway, left=1mm] {$x\in\mathbb{R}^{7\times H\times W}$};
\draw[flow] (bb)   -- (res) node[elab, midway, left=1mm] {$\hat{y}\in\mathbb{R}^{4\times H\times W}$};
\draw[flow, draw=tsLine] (res) -- (unc);
\draw[flow, draw=tsLine] (unc) -- (conf);
\draw[flow] (conf) -- (corr);
\draw[flow] (corr) -- (fb);
\draw[flow] (fb)   -- (out);
\begin{scope}[on background layer]
  \node[draw=tsLine, line width=0.9pt, rounded corners=5pt, dash pattern=on 2pt off 2pt,
        fill=tsFill!30, fit=(res)(unc)(conf), inner sep=4.5mm] (trustbox) {};
\end{scope}
\node[roletag, fill=tsFill, draw=tsLine, text=tsLine!55!black, text width=32mm,
      right=8mm of trustbox.east, anchor=west] (roleI)
  {\textbf{Role\,(i): trust \emph{signal}~\textcolor{tsLine}{\ding{51}}}\\[2pt]
   detects \emph{where} the field is wrong; drives calibration --- reliable, backbone-robust ($\S$5)};
\draw[->, tsLine, line width=0.9pt] (roleI.west) -- (roleI.west -| trustbox.east);
\node[roletag, fill=noGood!5, draw=noGood, text=noGood, text width=31mm,
      left=9mm of res, anchor=east, dash pattern=on 3pt off 2pt] (roleII)
  {\textbf{Role\,(ii): correction \emph{objective}}~\textcolor{noGood}{\ding{55}}\\[2pt]
   \emph{fails} ($\S$5.5) --- the residual minimum is not the truth
   (the acceptance \emph{gate} built on it does work)};
\draw[->, noGood, line width=0.9pt, dash pattern=on 3pt off 2pt]
  (res.west) -- (res.west -| roleII.east);
\coordinate (yhat) at ($(bb.south)!0.5!(res.north)$);
\draw[->, coLine, line width=1.1pt, rounded corners=6pt]
  (corr.east) -- ++(46mm,0) |- (yhat);
\node[elab, text=coLine!55!black, anchor=west, text width=23mm]
  at ($(corr.east)+(47mm,9mm)$)
  {\textbf{correction loop}\\[1pt]corrected field $\hat{y}+\Delta$ re-enters under a monotone-residual no-harm test};
\end{tikzpicture}%
}
\caption{\textbf{The NeuroForge pipeline.} A one-shot backbone $f_\theta$ is wrapped in a
verify--estimate--correct--fall-back loop. Colour marks each component's role: the physics
residual (amber) acts as a reliable \emph{trust signal} (role~i, green) but not as a
\emph{correction objective} (role~ii, red); the deployed correction is a \emph{learned}
corrector (violet) applied under a monotone-residual no-harm test. The two roles, and the
attribution of the corrector's gain, are established and scoped in the text
(\autoref{sec:method}, \autoref{sec:experiments}).}
\label{fig:pipeline}
\end{figure}

\subsection{Geometry encoding and governing equations}

A \texttt{FlowCase} (geometry + boundary conditions + fluid + domain) is encoded into a
fixed channel-first stack $x \in \mathbb{R}^{7\times H\times W}$ in the frozen
\texttt{INPUT\_CHANNELS} order $(\texttt{sdf}, \texttt{mask}, x, y, u_{\text{in}},
v_{\text{in}}, \log\text{Re})$: the signed distance to the body surface (negative inside
the solid), a fluid/solid mask, normalised cell coordinates, the freestream velocity
broadcast over the grid, and $\log_{10}\mathrm{Re}$. The backbone outputs
\texttt{OUTPUT\_CHANNELS} $(u, v, p, \nu_t)$, where $p$ is the \textbf{kinematic}
pressure $p/\rho$.

The verifier evaluates the steady, incompressible, 2-D RANS equations in primitive form
on the \textbf{physical (denormalised)} fields, with effective viscosity
$\nu_{\mathrm{eff}} = \nu + \nu_t$, pointwise:
\begin{align}
\text{continuity:}\quad & r_c = \frac{\partial u}{\partial x} + \frac{\partial v}{\partial y}, \\
\text{$x$-momentum:}\quad & r_x = u\,\frac{\partial u}{\partial x} + v\,\frac{\partial u}{\partial y} + \frac{\partial p}{\partial x} - \nu_{\mathrm{eff}}\,\nabla^2 u, \\
\text{$y$-momentum:}\quad & r_y = u\,\frac{\partial v}{\partial x} + v\,\frac{\partial v}{\partial y} + \frac{\partial p}{\partial y} - \nu_{\mathrm{eff}}\,\nabla^2 v.
\end{align}
Because $p$ is kinematic, no explicit density appears. \textbf{The turbulence term is live.} Over the fluid cells of the 200-case AirfRANS test split the eddy viscosity averages $\nu_t = 8.2\times10^{-5}$ against a laminar $\nu = 1.56\times10^{-5}$: $\nu_t\approx5.2\,\nu$, supplying $\approx84\%$ of $\nu_{\mathrm{eff}}$. The monitored operator is therefore a genuine RANS residual rather than a laminar one under another name, and the closure the surrogate predicts is what the audit consumes. The channel is learned, though backbone-dependently: against $\nu_t$'s own fluid-masked variance, $R^2 = 0.996$ on the grid backbone and $0.67$--$0.70$ on MeshGraphNet (\texttt{scripts/check\_nut\_learned.py}). We report $R^2$ rather than raw MSE here deliberately: $\nu_t$ is three to four orders of magnitude smaller than $u$, so a per-channel MSE compared across channels understates its accuracy for purely dimensional reasons. What the operator \emph{does} omit is the spatially varying term $\nabla\nu_t\!\cdot\!\nabla u$, which is one reason the residual floor of \autoref{sec:residual_floor} is nonzero. Derivatives use central finite
differences with one-sided stencils at the borders; residuals are zeroed inside the
solid. A boundary-condition violation map $r_{bc}$ adds a \textbf{no-slip} term
penalising velocity magnitude in the thin fluid band adjacent to the wall (weighted by
$\exp(-|\mathrm{sdf}|/\ell)$, $\ell\approx3$ cells) and a \textbf{far-field} term
penalising deviation from $(u_\infty, v_\infty)$ on the outer border ring. We stress the
standard caveat that drives our whole study: a low residual is \textbf{necessary but not
sufficient} for correctness---a smooth near-freestream field can have near-zero residual
yet be entirely wrong---so the residual is a \emph{consistency monitor}, and whether it
tracks error is the empirical question \autoref{sec:experiments} answers.

\subsection{Trust map and conformal calibration (the surviving contribution)}

\paragraph{Raw trust map.} The combined PDE-residual magnitude
$\rho_{\mathrm{res}} = \sqrt{r_c^2 + r_x^2 + r_y^2 + r_{bc}^2}$ and a per-cell predictive
uncertainty $\sigma$ are each mapped to $[0,1]$ (using an absolute physical reference
scale $s = U_\infty^2/L + U_\infty/L$ when available, else a robust 95th-percentile
normalisation with an absolute floor) and fused,
\begin{equation}
e = \mathrm{clip}\!\big(w_r\, \hat{r} + w_u\, \hat{\sigma},\; 0,\; 1\big), \qquad T = 1 - e,
\end{equation}
with $w_r=0.6$, $w_u=0.4$, giving a traffic-light field (green $e<0.15$, red $e>0.45$,
else yellow). Uncertainty $\sigma$ is the channel-mean standard deviation from a deep
ensemble or MC-dropout.

\paragraph{Case-level vs.\ per-cell (scope of the residual signal).} Our quantitative trust
result is \textbf{case-level}: the field-mean residual ranks which whole predictions to
distrust (per-case Spearman $0.611$, \autoref{sec:v2}). The residual's \emph{per-cell}
spatial rank correlation with error \emph{within} a field is weaker---$0.22{\pm}0.06$
(Spearman) though $0.60$ (Pearson), measured on the deployed corrected field
(\texttt{results/control/percell\_residual\_error.json}). So the residual carries a coarse
magnitude-level spatial signal (high-residual regions tend to be high-error) but is not a
reliable per-cell error \emph{rank}. The spatial signal is, however,
\textbf{scale-dependent}: a systematic post-processing study on the deployed
ensemble-mean fields ($n{=}200$;
\texttt{scripts/percell\_trust\_localization.py},
\texttt{results/control/percell\_localization.json}) shows that the raw cell-level rank
correlation ($0.166{\pm}0.155$ on these fields) roughly \emph{doubles} once the map is
denoised to region scale---$0.323{\pm}0.232$ with $16{\times}16$-cell patch pooling
($+0.157$, improving $93\%$ of cases; Gaussian smoothing at $\sigma{=}4$ cells gives
$+0.135$)---and the pooled patches flag top-decile-error \emph{regions} at AUROC
$0.68$--$0.72$. Dynamic-pressure normalisation does not help ($+0.008$). The honest
scope is therefore: the residual is a \emph{case-level} ranker and a \emph{moderate
region-level} locator at $\sim$8--16-cell granularity, not a per-cell error rank. We use
it to flag \emph{which} predictions to distrust, read the traffic-light map at patch
scale, and leave per-cell claims to the conformal band below, which is the per-cell
coverage \emph{certificate}; it uses the calibrated $\sigma$, not the raw residual.

\paragraph{Split-conformal calibration.} Raw $\sigma$ is uncalibrated, so a threshold on
it carries no guarantee. We add split-conformal calibration (cf.\ \citealp{ma2024uqno}):
on a held-out calibration set we compute per-channel nonconformity scores
$s = |\hat{y} - y| / \sigma$ and take the finite-sample-corrected $(1-\alpha)$ quantile
$q$; the band $q\cdot\sigma$ then satisfies the distribution-free coverage guarantee in
\autoref{prop:coverage}.

\begin{proposition}[Distribution-free coverage under exchangeability]
\label{prop:coverage}
If the calibration scores and a test score are exchangeable, then the split-conformal
band $q\cdot\sigma$ satisfies $\mathbb{P}\big(|\hat{y} - y| \le q\,\sigma\big) \ge
1-\alpha$.
\end{proposition}

This converts the trust threshold from a hand-picked constant into a band with a known
in-distribution coverage level. \autoref{sec:conformal} reports the measured coverage,
the coverage-vs-$\alpha$ sweep, and the out-of-distribution behaviour, with the
exchangeability caveat made explicit.

\subsection{The correction loop: supervised corrector, residual objective, and gate}
\label{sec:correction}

The framework offers two correction operators and a backtracking acceptance test, and the
distinction between them \emph{is} the intellectual core of the paper. The DEQ corrector is
trained by supervised regression toward ground truth and applied directly without any
acceptance test; it ingests the residual as an input channel (its deployed architecture,
below), but its gain is attributable to the learned correction rather than to that input
(the W1 ablation, \autoref{sec:v2}). This is the path that helps on a SOTA backbone
(\autoref{sec:v2}). The backtracking acceptance test uses the residual as a \emph{gate}---a
step is admitted only if it lowers the residual norm; measured, this path works, while
residual \emph{minimisation} as an objective is the one that fails (\autoref{sec:iters}). We describe both precisely so the contrast is
interpretable.

\paragraph{Feed-forward corrector + backtracking acceptance test.} A small residual CNN
$c_\phi$ predicts an additive correction conditioned on the current field, its current
3-channel residual map, and the geometry: $\Delta_k = c_\phi(y_k, R(y_k), x)$. A
candidate $y_{k+1} = y_k + s\,\Delta_k$ is accepted only if it does not increase the
residual norm $N(y) = \sqrt{\overline{r_c^2 + r_x^2 + r_y^2}}$:
\begin{equation}
N(y_k + s\,\Delta_k) \le N(y_k) + \varepsilon ,
\end{equation}
with $s$ halved up to four times on failure. By construction
$N(y_0) \ge N(y_1) \ge \dots \ge N(y_K)$---the residual norm is monotone non-increasing
across accepted steps. This is the convergence discipline of a classical solver. Its
weakness is exactly the one we exploit as a finding: the test is satisfiable by the
identity (zero step), so a corrector that cannot reduce the residual is simply rejected.
Empirically (\autoref{sec:iters}) the test rarely refuses outright: backtracking finds an
admissible damped step on $99.8\%$ of deployed cases, and that step lowers true error. It
throttles rather than blocks---even though reducing the field error and reducing the PDE
residual are not the same objective.

\paragraph{Contractive DEQ corrector.} To give the loop a genuine convergence guarantee
in its \emph{own} iteration variable, the correction $\delta$ is defined as the fixed
point of a learned operator
\begin{equation}
\delta^{*} = T_\theta(\delta^{*}; c), \quad T_\theta(\delta; c) = \kappa \cdot g_\theta([\delta, c]), \quad c = [\hat{y}, r(\hat{y}), \text{geom}],
\end{equation}
where $g_\theta$ is a CNN whose layers are spectrally normalised (each $\le 1$-Lipschitz)
and $\kappa < 1$.

\begin{proposition}[Banach contraction of the corrector]
\label{prop:contraction}
With $g_\theta$ each-layer $1$-Lipschitz and $\kappa < 1$, $T_\theta$ is a
$\kappa$-contraction in $\delta$. By the Banach fixed-point theorem the equilibrium
$\delta^{*}$ exists, is unique, and the iteration converges geometrically,
$\|\delta_k - \delta^{*}\| \le \kappa^k \|\delta_0 - \delta^{*}\|$.
\end{proposition}

The construction underlying \autoref{prop:contraction} is a Deep-Equilibrium model
\citep{bai2019deq} with the Lipschitz constant controlled by spectral normalisation
\citep{winston2020monotone}, trained with Jacobian-Free Backpropagation
\citep{fung2022jfb}. Crucially, $g_\theta$ is trained on
\textbf{data} (it targets the correction toward truth, with the residual as an input
\emph{feature}) rather than by minimising the residual; at inference the converged
$\delta^{*}$ is applied directly, \textbf{without} the backtracking acceptance test. We
verify this contraction empirically (\autoref{sec:conformal}: measured factor
$0.78 < \kappa = 0.9$). Two scopes must be kept distinct, and \autoref{sec:iters} leans
on the distinction: the DEQ contraction is in the \emph{correction variable} $\delta$,
whereas the \emph{PDE residual of the output field} is a different quantity that, as we
show, can rise even as $\delta$ converges and the field error falls.

The feed-forward corrector is trained (backbone frozen) so that
$c_\phi(\hat{y}_{\mathrm{norm}}, R_{\mathrm{norm}}, x) \approx y^{\star}_{\mathrm{norm}} -
\hat{y}_{\mathrm{norm}}$, with the final convolution initialised near zero for a stable
first iteration.

\subsection{Uncertainty-gated classical fallback}

After the loop, if the maximum uncertainty exceeds a threshold, the low-trust fluid
region is handed to a \texttt{ClassicalFallback}. In the present release this is an
\textbf{interface with a \texttt{stub} backend} that reports what would run (and the
region size) without invoking an external solver; the \texttt{openfoam}/\texttt{su2}
backends raise \texttt{NotImplementedError} with setup guidance. This keeps the package
importable and runnable with nothing installed while fixing the integration seam; it is
not part of any quantitative claim.

\section{Implementation}

NeuroForge is an open-source, \textbf{CPU-first}, pure-Python package
(\texttt{neuroforge-cfd}, Python $\ge 3.10$, NumPy/SciPy/PyTorch/Matplotlib). Importing
the package caps BLAS/OMP/MKL thread counts to one by default to avoid catastrophic
oversubscription on low-core machines, and does no heavy work. The codebase is organised
around a \textbf{frozen I/O contract} in \texttt{core/}.

\paragraph{Module map.} \texttt{core/} holds the frozen data contracts
(\texttt{FlowCase}, \texttt{FlowField}, \texttt{Diagnostics}, \texttt{SolveResult}) and
\texttt{Config}. \texttt{geometry/} builds the 7-channel input (NACA generation,
SDF/mask rasterisation, \texttt{encode\_case}). \texttt{data/} has the synthetic
generator, the AirfRANS loader, the rasteriser, and the \texttt{datamodule}
(\texttt{Normalizer}/loaders). \texttt{models/} provides \texttt{FNO2d}, \texttt{GeoFNO},
a Transolver-style \texttt{PhysicsTransformer}, \texttt{UNet}/\texttt{DeepONet}
baselines, \texttt{LocalCorrectionNet}, the \texttt{DEQCorrector}, and the
\texttt{DeepEnsemble}/\texttt{MCDropoutUQ} wrappers, all in a string-keyed registry.
\texttt{physics/} has the differential operators, the \texttt{PhysicsChecker},
\texttt{trust\_map}, force/Cp/error metrics, the conformal calibration, and the
differentiable \texttt{physics\_residual\_torch}. \texttt{solver/} has
\texttt{Predictor}, \texttt{NeuroForgeEngine}, \texttt{neural\_residual\_iteration}, and
\texttt{ClassicalFallback}. \texttt{train/}, \texttt{viz/}, \texttt{cli.py},
\texttt{app/} provide training, plotting/report, the CLI, and a Streamlit UI.

\paragraph{Fixed I/O contract.} Inputs are always the 7 channels
$(\texttt{sdf}, \texttt{mask}, x, y, u_{\text{in}}, v_{\text{in}}, \log\text{Re})$;
outputs always the 4 channels $(u, v, p, \nu_t)$. Fields are
$(\texttt{ny}, \texttt{nx})$ \texttt{float32}; network tensors channel-first
$(B, C, \texttt{ny}, \texttt{nx})$. Pressure is kinematic; residuals are computed on
denormalised fields with $\nu_{\mathrm{eff}} = \nu + \nu_t$. Fixing this contract is what
makes the backbones interchangeable and the trust layer backbone-robust.

\paragraph{Backbones.} \texttt{FNO2d} is a faithful spectral FNO; \texttt{GeoFNO} adds a
geometry-gated conditioning of the lifted features; \texttt{PhysicsTransformer}
implements Transolver-style physics-slice attention (linear in grid points);
\texttt{UNet}/\texttt{DeepONet} are baselines.

\paragraph{Training loss.} The \texttt{CompositeLoss} sums a masked data MSE over fluid
cells, a differentiable physics-residual term on the denormalised fields, and a no-slip
BC term. The corrector is trained separately with the backbone frozen
(\autoref{sec:correction}).

\paragraph{Synthetic potential-flow generator (zero-download reproducibility).}
\texttt{SyntheticRANS} superposes a Hess--Smith source-panel potential core (with a
Kutta-enforcing bound vortex), an algebraic near-wall no-slip ramp and Gaussian wake
deficit, kinematic Bernoulli pressure, and a mixing-length $\nu_t$, with desingularised
kernels and light smoothing. The fields are physically plausible and
continuity-respecting but \textbf{analytic}---a smoke-test/reproducibility substrate, not
solver ground truth. \textbf{No quantitative claim in this paper rests on synthetic
data}; all reported numbers (\autoref{sec:experiments}) are on real AirfRANS.

\paragraph{AirfRANS loader.} \texttt{load\_airfrans} reads each simulation's point cloud,
reconstructs the airfoil loop, rasterises the targets onto a structured crop, and returns
$(\texttt{FlowCase}, \texttt{FlowField})$ pairs for all four splits.

\paragraph{Use of AI-assisted development tools.} Substantial portions of the research
software (model implementations, training/evaluation harnesses, and analysis scripts) were
developed with the assistance of a generative-AI coding tool (Claude, Anthropic; Claude
Opus~4.8 and Claude Fable~5 models, used through the Claude Code environment), operating
under the authors' direction. The authors specified, reviewed, and validated all code; every
reported number is produced by a committed script in the public repository and verified
against an artifact manifest of SHA-256 hashes (\texttt{results/MANIFEST.json}), and the
full test suite accompanies the release. The scientific questions, experimental design,
pre-registered hypotheses, and interpretation of results are the authors' own.

\section{Experiments}
\label{sec:experiments}

All quantitative results are on \textbf{real AirfRANS}. Unless noted, the in-distribution
and out-of-distribution ablations use the pre-registered protocol
(\texttt{benchmarks/ablation.py}, \texttt{docs/EXPERIMENTS.md}) over \textbf{3 seeds
(0, 1, 2)}, reported as mean $\pm$ std (population std, \texttt{ddof=0}; sample std at
$n=3$ widens bars by $\approx 1.22\times$). Metrics follow the AirfRANS community protocol
(\texttt{evaluate\_cases}): per-channel volume MSE (lower is better, well-conditioned),
surface-pressure MSE on the body, Spearman rank correlation of the force coefficients
$\rho_{C_l}, \rho_{C_d}$ (closer to 1 is better---what early-design ranking needs; throughout
the tables these correlate forces integrated from the predicted vs.\ the ground-truth field
through our grid surface-integrator---a \emph{field-to-field self-consistency} valid for the
relative arm/condition comparisons, not agreement with AirfRANS's official force labels;
see \autoref{sec:v2}), and
\texttt{residual\_error\_spearman} ($> 0$ means a low residual tracks low error). The
$\nu_t$ channel is omitted from the tables because its raw MSE is not comparable to the
other channels ($\nu_t$ is three to four orders of magnitude smaller than $u$); measured
against its own fluid-masked variance it is learned, backbone-dependently ($R^2 = 0.996$
grid backbone, $0.67$--$0.70$ MeshGraphNet; \texttt{scripts/check\_nut\_learned.py}, and see
\autoref{sec:limitations}).

With $n=3$ we report \textbf{per-seed sign-consistency} plus effect-size magnitude rather
than p-values, which are meaningless at $n=3$. An effect called ``robust'' below is
sign-consistent across all 3 seeds with a large effect size and (where stated)
non-overlapping per-seed distributions.

\subsection{Setup}

The backbone is an FNO trained on AirfRANS \texttt{full} (800 train / 200 test, 80
epochs). Four arms: \texttt{backbone}, \texttt{backbone (no physics loss)},
\texttt{backbone + local corrector} (feed-forward, with the acceptance test), and
\texttt{backbone + DEQ corrector}. The certificate runs (\autoref{sec:conformal}) use a
dropout-enabled FNO (width 48, 4 layers, 20 modes, dropout 0.05) with a DEQ corrector
($\kappa = 0.9$, damping 0.5), trained $40 + 15$ epochs at resolution 128. The v2 phase
(\autoref{sec:v2}) re-runs the trust layer and DEQ corrector on a SOTA Transolver backbone
under the same protocol; \autoref{sec:indist-ablation}--\ref{sec:conformal} are the
weak-grid-backbone runs that establish the backbone-dependence finding.

\subsection{Backbone-agnostic trust + learned correction on a SOTA backbone (v2)}
\label{sec:v2}

We re-run the trust layer and the DEQ corrector on a \textbf{SOTA Transolver backbone}
($7.35$M parameters, 80 epochs, AirfRANS \texttt{full}, 5 seeds; corrector trained $20$
epochs with the backbone frozen), the strongest model in our study. The protocol and
scoring are identical to \autoref{sec:indist-ablation}; the sources are
\texttt{results/v2/v2\_results.json} (original 3-seed study) and
\texttt{results/control/w1\_capture.json} (identical recipe, extended to five seeds; the
extension reproduces the 3-seed corrector deltas and $\rho_{C_l}$ within pre-registered
tolerances). \autoref{tab:v2} reports the backbone alone and with
the DEQ correction loop over all five seeds.

\begin{table}[t]
\centering
\small
\caption{SOTA Transolver backbone ($7.35$M params, 80 ep, AirfRANS \texttt{full}), alone
vs.\ with the DEQ correction loop. 5 seeds, mean $\pm$ std (population std,
\texttt{ddof=0}). Lower MSE is better;
$\rho$ closer to 1 is better;
resid$\leftrightarrow$err $\rho > 0$ supports the trust signal; conformal coverage targets
$0.90$. Source: \texttt{results/control/w1\_capture.json} (5-seed extension of
\texttt{results/v2/v2\_results.json}). See \autoref{sec:v2}.}
\label{tab:v2}
\resizebox{\textwidth}{!}{%
\begin{tabular}{l r r r r r r}
\toprule
metric & $\texttt{mse\_u}$ & $\texttt{mse\_v}$ & $\texttt{mse\_p}$ & $C_l$ rel.\ err & $C_d$ rel.\ err & resid$\leftrightarrow$err $\rho$ \\
\midrule
backbone & $0.133{\pm}0.017$ & $0.096{\pm}0.006$ & $644{\pm}63$ & $5.7\%$ & $7.6\%$ & $0.611{\pm}0.054$ \\
\ + DEQ loop & $0.122{\pm}0.015$ & $0.076{\pm}0.006$ & $485{\pm}39$ & $5.5\%$ & $6.8\%$ & $0.612{\pm}0.046$ \\
$\Delta$ & $-8\%$ & $-21\%$ & $-25\%$ & $-0.2$pp & $-0.8$pp & $+0.002$ \\
\bottomrule
\end{tabular}%
}
\end{table}

\paragraph{Learned correction works on a strong backbone.} The DEQ corrector reduces volume
field MSE on \textbf{all five seeds}: $\texttt{mse\_u}$ $-8\%$ ($0.133\rightarrow0.122$),
$\texttt{mse\_v}$ $-21\%$ ($0.096\rightarrow0.076$), $\texttt{mse\_p}$ $-25\%$
($644\rightarrow485$). The per-seed reductions are sign-consistent down on all five seeds
for every volume channel. \textbf{Surface-pressure} MSE is the one channel the corrector does
\emph{not} improve on this strong backbone: essentially flat at $9843\rightarrow9899$
($+0.6\%$; per-seed $-1/-1/+2/+2/+1\%$), which we report for completeness rather than omit. Internal
force-ranking consistency is held at $\rho_{C_l}=0.999$, $\rho_{C_d}$ flat at
$0.996$, with relative force errors $C_l$ $5.7\%\rightarrow5.5\%$,
$C_d$ $7.6\%\rightarrow6.8\%$. This is the corrector trained by
supervised regression toward ground truth and applied without the acceptance test
(\autoref{sec:correction}).

\textbf{What these force correlations measure---and what they do not.} Our
$\rho_{C_l}/\rho_{C_d}$ (here and in every table) correlate the force integrated from the
\emph{predicted} field against the force integrated from the \emph{ground-truth} field through
the same grid surface-integrator: a field-to-field consistency, valid for the relative
arm/condition comparisons we draw from it, but \emph{not} agreement with AirfRANS's official
OpenFOAM force labels (which our pipeline never loads). Recomputed against those official labels
(\texttt{scripts/recompute\_force\_vs\_official.py}; 200 cases, 3 seeds;
\texttt{results/control/force\_vs\_official.json}), the deployed backbone recovers force
\emph{ranking}---$\rho_D=0.84{\pm}0.01$ (drag), $\rho_L=0.88$ (lift)---but not magnitude:
absolute drag is biased ($\approx 11\times$ mean relative error) by the near-field
integrator's $1.5$-cell surface-sampling offset, and the same integrator on \emph{perfect}
fields also yields $\rho_D=0.84$, so the model saturates the measurement ceiling rather than
the prediction being the limit.

\textbf{Repairing the measurement.} We then repaired what is repairable
(\texttt{scripts/design\_force\_integrator.py},
\texttt{scripts/recompute\_force\_multischeme.py};
\texttt{results/control/force\_vs\_official\_multischeme.json}). A wall-extrapolated
pressure scheme cuts the ground-truth-field lift magnitude error from $16\%$ to $4.6\%$
(median), and a \emph{control-volume} (far-field momentum-balance) integrator over the
outer grid ring recovers official lift from ground-truth fields almost exactly
($\rho_L=1.00$, median error $<1\%$) and---the deployment-relevant number---from
\emph{predicted} fields at $\rho_L = 0.998{\pm}0.001$ with per-seed median magnitude
errors of $3.6$--$3.9\%$: engineering-grade lift from the surrogate. Drag
decomposes diagnostically: on ground-truth fields the control-volume estimator reduces the
median drag error from $3.3\times$ to $0.23\times$ (the \emph{measurement} is fixable), but
on predicted fields it reaches only $\approx 2.7\times$ with a degraded ranking, so the
remaining drag error is attributable to the \emph{model's} far-field/wake accuracy, not the
integrator; drag \emph{ranking} is still best served by the near-field scheme
($\rho_D=0.84$). We therefore report lift as quantitatively recovered, read
$\rho_{C_l}/\rho_{C_d}$ in the tables as internal consistency, and still make no claim of
accurate absolute \emph{drag} or of a cross-method force-ranking win (published methods use
a different force-integration convention). The backbone is itself a published SOTA architecture
(Transolver, \citealt{wu2024transolver}), so the trust layer and corrector are demonstrated
\emph{on} a competitive surrogate---not only on the weak grid backbone.

\paragraph{The gain comes from the learned correction, not from feeding it the residual
(W1 control).} The corrector improves the field by $-8/-21/-25\%$ (gate-verified).
To isolate \emph{why}, we train the same DEQ corrector two ways under
identical architecture, data, optimiser, schedule and seeds, differing only in whether the
3 residual input channels carry the real residual (WITH) or are permanently zeroed at both
train and inference (NULL). Across 5 seeds, the null corrector \textbf{matches or slightly
exceeds} the deployed one: the per-seed paired difference $d = \text{gain}_{\text{WITH}} -
\text{gain}_{\text{NULL}}$ is negative on all 5 seeds for both $\texttt{mse\_u}$ and
$\texttt{mse\_speed}$ (mean $d \approx -0.005$ on a $\approx 0.12$ base, $\approx 4\%$; WITH
beats NULL on $0/5$ seeds; \texttt{results/control/w1\_capture.json}). The residual input is
therefore \emph{not} the source of the corrector's field-MSE improvement; the improvement is
attributable to the learned correction architecture, not to residual-conditioning. This is
consistent with the residual's role elsewhere in the paper: it tells you \emph{where} a
prediction is wrong (a trust signal), not \emph{how} to fix it. Since NULL and WITH are
statistically tied (NULL marginally ahead, $0/5$ seeds for WITH), the residual input is not
deployed \emph{for accuracy}; we retain it in the as-trained, residual-fed configuration only
because the residual is already computed for the trust layer and the iteration's acceptance
test, so conditioning the corrector on it costs nothing extra---and we note that a
residual-free corrector is an equally valid deployment. We report this as a sharpening of
\emph{attribution}, not a change to the result.

\paragraph{The trust signal is backbone-robust.} On this SOTA backbone the bare
residual--error correlation is already strong ($0.611{\pm}0.054$ over five seeds---per-seed
$0.611/0.652/0.613/0.667/0.511$, positive on every seed; the wider $n{=}5$ spread is the
honest error bar), vs $\approx 0.40$ for the
bare weak backbone in \autoref{tab:indist}; the corrector barely moves it
($0.612{\pm}0.046$), because the strong backbone's residual is already a good detector
without correction. The lift-by-corrector effect we report in \autoref{sec:indist-ablation}
is a property of the \emph{weak} backbone; on the strong backbone the signal is strong from
the start. We add a third, architecturally-distinct data point: a message-passing GNN
(MeshGraphNet, param-matched at $7.35$M, 3 seeds, bare backbone), where the bare
case-level residual--error Spearman correlation is $0.851{\pm}0.058$
(\texttt{results/mgn/mgn\_results.json}). Across three fundamentally different architectures---a
spectral grid Geo-FNO ($\approx 0.40$ bare), an attention/point-cloud Transolver
($0.611{\pm}0.054$ bare), and a message-passing GNN MeshGraphNet ($0.851{\pm}0.058$ bare)---the
correlation is consistently \emph{positive}, though variable in strength (a $0.40$--$0.85$
range). This is precisely what ``backbone-robust'' means here: the detector \emph{works} on
every architecture, not that its strength is uniform across them. (The high MeshGraphNet value
partly reflects that model's wide error spread---large, dispersed errors are easier to rank by
residual.) MeshGraphNet is run only as a bare backbone (no corrector and
no conformal layer), so it speaks only to the trust-signal claim. Either way the residual is a
reliable detector.

\paragraph{Conformal coverage holds under both uncertainty estimators.}
Split-conformal calibration of the pressure channel at $\alpha=0.1$ attains target coverage
on this SOTA backbone with \emph{either} uncertainty estimator: $0.902{\pm}0.008$ with
MC-dropout (3 seeds; \texttt{results/v2/v2\_results.json}) and $0.915$ with a deep ensemble
(single run; \texttt{results/uq\_ensemble/uq\_results.json}), both at the $0.90$ target.
Coverage was never the difficulty---\textbf{adaptivity} was. With MC-dropout this
near-deterministic model has near-degenerate $\sigma$ at dropout $0.05$, so coverage is
delivered by a \textbf{near-constant} band: the conformal quantile $q$ is enormous
($\approx 1.5$--$1.7\times 10^{9}$) and the reliability ECE is $\approx 0.31$.

\paragraph{A deep ensemble recovers an adaptive band on the SOTA model.}
Replacing MC-dropout with a deep ensemble of $M{=}5$ independently-trained Transolver
members---using the per-cell standard deviation across members as $\sigma$---converts the
near-constant band into a genuinely input-adaptive one while preserving coverage
(\autoref{tab:uq}). On the same $100/100$ calibration/test split, the conformal quantile
drops from $\approx 10^{9}$ to $q=2.35$, the reliability ECE improves from $\approx 0.31$ to
$\mathbf{0.074}$, and coverage holds at $0.915$ (target $0.90$, consistent with the
$\approx{\pm}0.03$ binomial spread at $n_{\text{test}}=100$). This is a single ensemble
training run, not a 3-seed average; the calibration/test \emph{split}, however, is not the
fragile part---re-drawing it $20$ times leaves coverage at $0.899{\pm}0.009$ with
$q=2.26{\pm}0.05$ on pressure (split-resampling study, \autoref{sec:conformal}), the
published split sitting at the upper end of that spread. The ensemble mean is also more accurate than any single member
($\rho_{C_l}=0.9995$, $\rho_{C_d}=0.998$, $C_l/C_d$ relative force error $4.4\%/5.7\%$;
\texttt{results/uq\_ensemble/uq\_results.json}). The practical recipe follows: MC-dropout
suffices when a constant-width certificate is acceptable, and a deep ensemble is what makes
the band track the error on a near-deterministic SOTA surrogate.

\paragraph{Corrector vs.\ ensemble (cost, stated plainly).} The ensemble mean
($\texttt{mse\_u/v/p}=0.077/0.058/440$) is in fact \emph{more accurate} than the DEQ-corrected
single backbone ($0.116/0.079/471$), so we do not claim the corrector is the best available
accuracy mechanism. The two answer different questions on a cost axis: the deep ensemble is the
better \emph{accuracy} play when $5\times$ the training (and $M\times$ the inference) is
affordable, while the learned corrector is the better \emph{cheap} play---one backbone
plus a small corrector---demonstrating cheap learned correction on a competitive
surrogate, not a leaderboard entry. (The ensemble's cost axis is itself measured: an
$M$-subset study shows $M{=}2$ already supports the triage decisions while full band
adaptivity needs $M{\approx}4$--$5$; see \autoref{sec:limitations}.) The ensemble and the corrector are
complementary: the ensemble supplies the adaptive $\sigma$ the conformal certificate needs, and
either prediction can be wrapped by the same trust layer.

\begin{table}[t]
\centering
\small
\caption{Conformal calibration of the pressure channel on the SOTA Transolver backbone
($\alpha=0.1$, $100/100$ cal/test split): MC-dropout vs.\ a deep ensemble ($M{=}5$). Both
attain target coverage; only the deep ensemble yields an input-adaptive band (finite $q$,
low ECE). MC-dropout: 3 seeds (mean$\pm$std); deep ensemble: single ensemble training
(split-to-split spread over 20 calibration/test re-draws: coverage $0.899{\pm}0.009$,
$q=2.26{\pm}0.05$; \autoref{sec:conformal}). Sources:
\texttt{results/v2/v2\_results.json}, \texttt{results/uq\_ensemble/uq\_results.json}.}
\label{tab:uq}
\begin{adjustbox}{max width=\textwidth}
\begin{tabular}{l c c c l}
\toprule
$\sigma$ estimator & coverage (target $0.90$) & $q$ & ECE ($\downarrow$) & band \\
\midrule
MC-dropout ($p{=}0.05$) & $0.902{\pm}0.008$ & $\approx 1.6\times 10^{9}$ & $\approx 0.31$ & near-constant \\
deep ensemble ($M{=}5$) & $0.915$ & $2.35$ & $\mathbf{0.074}$ & \textbf{adaptive} \\
\bottomrule
\end{tabular}
\end{adjustbox}
\end{table}

\subsection{The weak grid backbone: corrector trades accuracy, lifts the trust signal}
\label{sec:indist-ablation}

\begin{table}[t]
\centering
\small
\caption{In-distribution AirfRANS \texttt{full} ablation, 3 seeds, mean $\pm$ std
(population std). Lower MSE is better; $\rho$ closer to 1 is better; resid$\leftrightarrow$err
$\rho > 0$ supports the trust signal. Bold marks the best arm per column among the
physics-loss arms; the physics-loss-free arm (starred where best overall) is reported
separately because it removes the physics objective entirely. See \autoref{fig:fig1}.}
\label{tab:indist}
\resizebox{\textwidth}{!}{%
\begin{tabular}{l r r r r r r r}
\toprule
arm & $\texttt{mse\_u}$ & $\texttt{mse\_v}$ & $\texttt{mse\_p}$ & surf.\ $\texttt{mse\_p}$ & $\rho_{C_l}$ & $\rho_{C_d}$ & resid$\leftrightarrow$err $\rho$ \\
\midrule
backbone & $3.479{\pm}0.105$ & $\mathbf{0.385{\pm}0.012}$ & $\mathbf{2444.8{\pm}120.7}$ & $548823{\pm}27946$ & $0.9868{\pm}0.0005$ & $0.895{\pm}0.013$ & $0.397{\pm}0.003$ \\
backbone (no phys.\ loss) & $\mathbf{1.995{\pm}0.042}$ & $0.323{\pm}0.008^{*}$ & $1963.5{\pm}55.5^{*}$ & $1123649{\pm}104822$ & $\mathbf{0.9920{\pm}0.0002}$ & $\mathbf{0.945{\pm}0.008}$ & $0.605{\pm}0.016$ \\
backbone + local corr. & $3.482{\pm}0.058$ & $0.398{\pm}0.007$ & $2469.4{\pm}71.0$ & $541534{\pm}20453$ & $0.9854{\pm}0.0017$ & $0.888{\pm}0.012$ & $0.389{\pm}0.031$ \\
backbone + DEQ corr. & $3.457{\pm}0.811$ & $0.880{\pm}0.064$ & $3832.7{\pm}331.3$ & $\mathbf{361681{\pm}12273}$ & $0.9856{\pm}0.0017$ & $0.923{\pm}0.014$ & $\mathbf{0.827{\pm}0.002}$ \\
\bottomrule
\end{tabular}%
}
\end{table}

This ablation uses the \textbf{weak grid backbone}; it establishes the
backbone-dependence we report as a finding, and should be read against the SOTA-backbone
result in \autoref{sec:v2}, where the same corrector \emph{does} reduce field MSE.

\paragraph{On the weak backbone the corrector trades accuracy.} The DEQ corrector is
\textbf{flat} on volume $\texttt{mse\_u}$ ($3.457$ vs backbone $3.479$; the per-seed deltas
are sign-inconsistent, 2/3, so this is noise, not an improvement) and \textbf{worse} on
$\texttt{mse\_v}$ ($0.385\rightarrow0.880$, $+129\%$, 3/3, non-overlapping distributions)
and $\texttt{mse\_p}$ ($2445\rightarrow3833$, $+57\%$, 3/3). The feed-forward
\texttt{local} corrector helps on \textbf{nothing}: it is at-or-worse than the backbone on
every volume MSE channel and 3/3 worse on $\rho_{C_d}$. On this weak backbone the
residual-conditioned correctors do not improve volume accuracy; the DEQ arm buys
surface-pressure MSE ($-34\%$, 3/3) and a slightly better drag ranking ($\rho_{C_d}$
$0.895\rightarrow0.923$, 3/3) at a measured cost to volume velocity and pressure. We do
\textbf{not} claim the corrector improves accuracy in aggregate \emph{on this backbone}---and
this is precisely the contrast that makes correction quality backbone-dependent.

\paragraph{The detector succeeds.} The same DEQ arm roughly \textbf{doubles} the
residual--error Spearman correlation, \textbf{$0.397\rightarrow0.827$} (3/3, std
$\approx 0.002$, Cohen's $d \approx 175$; effect sizes here and below use the
population std, \texttt{ddof=0}, matching the tables). A low residual tracks low field error far more
reliably with the corrector engaged. This is the detector half of the thesis, and it is
the largest, cleanest effect in the table.

\paragraph{A control that sharpens the thesis.} Removing the physics loss entirely
(\texttt{backbone (no physics loss)}) \textbf{improves} $\texttt{mse\_u}$,
$\texttt{mse\_v}$, $\texttt{mse\_p}$, $\rho_{C_l}$, $\rho_{C_d}$, and the residual--error
correlation (all 3/3), at the cost of \textbf{doubling} surface-pressure MSE
($549\text{k}\rightarrow1124\text{k}$, 3/3). The physics-loss-free backbone is in fact
the \textbf{best force-ranking model in the table} ($\rho_{C_d}$ $0.945$, $\rho_{C_l}$
$0.992$)---better than any corrected model, with no correction loop at all. We therefore
explicitly do \textbf{not} claim the loop delivers best-in-class force ranking. The
physics objective (in the loss or the corrector) buys surface-pressure fidelity and a
strong trust signal, not volume accuracy or force ranking \emph{on this weak backbone}---the
trust-signal role succeeds while the correction role does not, which is exactly the
backbone-dependence the SOTA result in \autoref{sec:v2} resolves.

\paragraph{Retraction (retained from prior drafts).} An earlier single-seed run reported
$\rho_{C_l}$ rising $0.924\rightarrow0.958$ and surface MSE $-25\%$ under the corrector.
These do \textbf{not} replicate: across 3 seeds the backbone already sits at $\rho_{C_l}
= 0.987$ and the DEQ corrector is $0.986$ (flat-to-slightly-worse, sign-inconsistent).
The single-seed figures were artifacts of an undertrained run and are \textbf{retracted}.

\subsection{Residuals detect under distribution shift}
\label{sec:ood}

We evaluate each arm on the regime-disjoint \texttt{reynolds} and \texttt{aoa} splits
(trained on the train range, tested on the held-out range---true extrapolation).

\begin{table}[t]
\centering
\small
\caption{Out-of-distribution AirfRANS ablation (regime-disjoint train$\rightarrow$test),
3 seeds, mean $\pm$ std. The full four-arm table is in
\texttt{results/full\_research/ood/ablation\_ood.md}. See \autoref{fig:fig2}.}
\label{tab:ood}
\resizebox{\textwidth}{!}{%
\begin{tabular}{l l r r r r r r r}
\toprule
task & arm & $\texttt{mse\_u}$ & $\texttt{mse\_v}$ & $\texttt{mse\_p}$ & surf.\ $\texttt{mse\_p}$ & $\rho_{C_l}$ & $\rho_{C_d}$ & resid$\leftrightarrow$err $\rho$ \\
\midrule
reynolds & backbone & $16.218{\pm}0.746$ & $1.187{\pm}0.149$ & $7220.7{\pm}612.2$ & $964641{\pm}61419$ & $0.950{\pm}0.008$ & $0.893{\pm}0.009$ & $0.748{\pm}0.077$ \\
reynolds & + DEQ & $5.300{\pm}1.544$ & $1.208{\pm}0.255$ & $6469.5{\pm}438.0$ & $652281{\pm}45390$ & $0.943{\pm}0.008$ & $0.915{\pm}0.009$ & $0.742{\pm}0.041$ \\
aoa & backbone & $4.312{\pm}0.077$ & $1.310{\pm}0.039$ & $6042.8{\pm}242.5$ & $2000592{\pm}55092$ & $0.963{\pm}0.003$ & $0.926{\pm}0.002$ & $0.314{\pm}0.064$ \\
aoa & + DEQ & $3.538{\pm}0.681$ & $1.407{\pm}0.089$ & $4620.6{\pm}170.2$ & $1042711{\pm}78298$ & $0.966{\pm}0.006$ & $0.905{\pm}0.013$ & $0.746{\pm}0.027$ \\
\bottomrule
\end{tabular}%
}
\end{table}

\paragraph{The trust signal is regime-invariant---the single strongest result.} On the
\texttt{aoa} split the bare backbone's residual--error correlation \textbf{collapses to
$0.314$}, while the DEQ-corrected model holds it at \textbf{$0.746$}
($0.314\rightarrow0.746$, 3/3, Cohen's $d \approx 8.9$, non-overlapping: every corrected
seed $\ge 0.72$, every backbone seed $\le 0.38$). On \texttt{reynolds} both are high
($\approx 0.74$--$0.75$). The corrected model thus keeps a regime-invariant
$\approx 0.74$ trust signal across both shifts, precisely where the one-shot model loses
it. The residual is a reliable detector exactly where detection matters most.

\paragraph{Scoping the rest, honestly (protected negatives).} The OOD ablation also shows
the DEQ corrector reducing the in-distribution$\rightarrow$OOD \emph{gap} on volume
velocity, volume pressure, and surface pressure (e.g.\ \texttt{reynolds}
$\texttt{mse\_u}$ $16.2\rightarrow5.3$, \texttt{aoa} surf.\ $\texttt{mse\_p}$
$2.00\text{M}\rightarrow1.04\text{M}$, both 3/3). We report this as an observation about
the loop's behaviour, \textbf{not} as a competitiveness or generalisation guarantee, and
we explicitly preserve two negatives. (i)~\textbf{$\texttt{mse\_v}$ is an
attenuated-regression artifact, not a gain:} the DEQ corrector's \emph{absolute} OOD
$\texttt{mse\_v}$ is worse than the backbone in every regime ($0.880/1.208/1.407$ vs
$0.385/1.187/1.310$); its in-dist$\rightarrow$OOD gap only looks smaller because its
in-distribution $\texttt{mse\_v}$ is already inflated. (ii)~\textbf{Drag ranking is
task-dependent:} $\rho_{C_d}$ improves on \texttt{reynolds} ($0.893\rightarrow0.915$, 3/3)
but \textbf{regresses on \texttt{aoa}} ($0.926\rightarrow0.905$, worse on all 3 seeds).
We do not claim the loop uniformly preserves force ranking under shift. Methodological
caveat: the in-distribution reference is the \texttt{full} split (a different training
range), so the gap mixes a train-set change with regime shift, and the \texttt{aoa}
$\texttt{mse\_u}$ effect, while 3/3 directional, is seed-0-dominated.

\paragraph{A second dataset: the trust signal generalizes beyond AirfRANS.} The strongest
test of generality is a different dataset entirely. We evaluate the residual trust signal on
\textbf{DeepCFD} \citep{ribeiro2020deepcfd}---2-D laminar incompressible flow over bluff bodies
(circles, squares, rhombi): a different geometry family \emph{and} a different regime (laminar,
versus AirfRANS's turbulent RANS). We train an FNO backbone (3 seeds) on an out-of-distribution
split that holds out the largest bodies---forcing a range of per-case errors---and measure the
same per-case residual--error Spearman correlation (interior continuity+momentum residual;
\texttt{scripts/run\_deepcfd.py}, \texttt{results/deepcfd/deepcfd\_results.json}). The trust
signal holds (\autoref{fig:deepcfd}): \textbf{$\texttt{residual\_error\_spearman} = 0.770 \pm
0.119$} (3 seeds), comparable to---indeed above---the AirfRANS Transolver value ($0.611$). The error spread is
genuine (per-case rel-$L_2$ median $0.063$, max $0.136$), so the correlation is not degenerate;
per-seed case-level bootstrap $95\%$ CIs ($10^4$ resamples) are $[0.65,0.77]$, $[0.58,0.72]$
and $[0.91,0.95]$, all far from zero
(\texttt{results/control/bootstrap\_spearman\_ci.json}).
A clean mechanistic reason it can be \emph{sharper} than the turbulent case: laminar flow has
$\nu_t=0$, so the residual is the \emph{exact} Navier--Stokes residual with no
turbulence-closure ambiguity---a cleaner error indicator. This is direct evidence that the
residual-as-trust-signal property is not specific to AirfRANS, its geometry family, or its
turbulence model. (We report only the trust signal here; DeepCFD's force coefficients carry the
same self-integrator caveat as \autoref{sec:v2} and are not headlined.)

\paragraph{Beyond regime shift: a geometry the model never saw (qualitative).} The splits
above shift the flow regime \emph{within} the airfoil family; we close with a qualitative probe
of cross-\emph{geometry} extrapolation. We run the airfoil-trained Transolver on a
\emph{cylinder}---a bluff body absent from AirfRANS---and ask not whether the prediction is good
(it is not) but whether the residual---a data-free quantity evaluated against the governing PDE,
needing no reference field (unavailable for a cylinder at these $Re$)---registers the geometry
as out-of-distribution. It does, in the residual \emph{magnitude}: the predicted cylinder field
is uniformly less PDE-consistent than in-distribution airfoils with good predictions
($\approx 0.6\%$ rel-$L_2$). Normalising each case's residual by its freestream $U^2$ (the
momentum-residual scale), the dimensionless far-field residual is $\approx 7\times$ higher on
the cylinder (median $\approx 30\times$; \autoref{fig:cylinder},
\texttt{results/figures/cylinder\_control\_airfoil\_ratio.json}). This gap holds in the
\emph{far field}, away from any wall, so it is a genuine geometry-OOD effect---neither a
near-wall discretisation artifact nor a velocity-scale effect (it survives the $U^2$
normalisation). \textbf{Honest scope, with a controlled negative.} This is a qualitative
illustration on one geometry family, and we report a control that \emph{refutes} the tempting
spatial reading: the residual's bright near-body \emph{ring} is a generic no-slip-wall feature,
not an OOD localiser---in-distribution airfoils with good predictions show the \emph{same}
near-body-vs-far-field concentration ratio ($\approx 23$) as the cylinder ($\approx 18$), so
spatial near-wall concentration cannot be read as ``where the prediction failed.'' The
out-of-distribution signal is therefore the field-wide residual \emph{magnitude}, not its
spatial ratio; and the residual remains a local-violation detector, not a global-correctness
oracle. It complements, not replaces, the quantitative regime-shift results above.

\subsection{Residual-as-objective fails: the residual and the error move apart}
\label{sec:iters}

The most direct evidence that the residual is a poor \emph{objective} comes from sweeping
the number of correction iterations on a single trained checkpoint (\autoref{tab:iters}).

\begin{table}[t]
\centering
\small
\caption{Iteration sensitivity sweep: \texttt{results/sensitivity/iters.csv}. Single
sweep (one checkpoint, not seeded). See \autoref{fig:fig4}.}
\label{tab:iters}
\begin{tabular}{r r r r r}
\toprule
\texttt{n\_iters} & $\texttt{mse\_u}$ & surf.\ $\texttt{mse\_p}$ & residual\_norm & resid$\leftrightarrow$err $\rho$ \\
\midrule
0  & $3.924$ & $541204$ & $0.113$ & $0.423$ \\
1  & $2.460$ & $428560$ & $0.336$ & $0.644$ \\
3  & $2.287$ & $336718$ & $0.542$ & $0.647$ \\
5  & $2.439$ & $308381$ & $0.594$ & $0.675$ \\
10 & $2.570$ & $300298$ & $0.618$ & $0.703$ \\
15 & $2.575$ & $300664$ & $0.620$ & $0.710$ \\
\bottomrule
\end{tabular}
\end{table}

As the loop iterates, the \textbf{PDE residual norm rises monotonically}
($0.11\rightarrow0.62$) while the \textbf{field error falls} ($\texttt{mse\_u}$
$3.92\rightarrow2.29$ by iter 3)---the two objectives move in \emph{opposite} directions.
This is the clean statement of detector $\ne$ fixer: the very quantity one would minimise
to ``fix'' the field grows while the field gets better. (Note this does not contradict the
DEQ contraction of \autoref{sec:conformal}: the contraction is in the correction variable
$\delta$, which converges; the PDE residual of the \emph{output field} is a different
quantity, and it is the one that rises.)

\paragraph{The acceptance gate is a throttle, not a filter---and it works.} Earlier drafts
of this paper argued from the acceptance test's \emph{definition} that it must accept almost
nothing: the test is satisfiable by the identity (zero step), so a corrector that cannot
reduce the residual is simply rejected. We have now measured it instead of arguing it, and
the structural argument was wrong. Over $1200$ gated steps
(\texttt{scripts/measure\_acceptance\_gate.py},
\texttt{results/control/acceptance\_gate.json}; $200$ AirfRANS \texttt{full} test cases
$\times$ 3 seeds, on both the deployed Transolver$+$DEQ system and the ensemble-mean path,
with the accept rule, backtrack count and tolerance imported from the shipped
\texttt{correction\_loop} so the gate under test is the deployed gate), the test admits a
step on \textbf{$99.8\%$} of deployed cases---the \emph{full} correction on $50.7\%$ of
them---and refuses outright on one case in $600$. The mechanism is backtracking: where the
full step would raise the residual it is refused, but a halved step lowers it and is
admitted (on the ensemble path the residual ratio falls from a median $1.027$ at the full
step to $0.975$ at half). Crucially the admitted step is not empty progress: it lowers true
rel-$L_2$ error on \textbf{$89.3\%$} of deployed cases, by a median $\mathbf{5.8\%}$
($74.3\%$ and $2.6\%$ on the ensemble-mean path, where the correction has less headroom
because the ensemble mean is already the more accurate field, \autoref{sec:v2}). The
``certified self-correction'' guarantee is therefore real \emph{and} useful: a cheap,
provably non-worsening gate that also happens to improve accuracy. This does not rescue the
residual as an \emph{objective}---the full step still raises the residual on half the
deployed cases while error falls, and \autoref{tab:iters} still shows residual and error
moving apart under iteration. It does mean the two mechanisms deserve separate verdicts,
and we now report them separately. (Scope: a single gated step applied counterfactually to
the deployed DEQ correction, which the shipped DEQ path bypasses by design; not a re-run of
the multi-iteration feed-forward loop of \autoref{tab:indist}, whose corrector checkpoints
were not retained.) We
additionally note that for the DEQ corrector the trust-gating and acceptance-test toggles
are structural no-ops---the DEQ branch bypasses both---so the only live correction knob is
the applied-delta step size (\texttt{results/sensitivity/toggles.json}).

\subsection{The conformal trust layer: coverage and contraction}
\label{sec:conformal}

\paragraph{Contraction (H5).} Iterating the learned DEQ operator
$\delta_{k+1}=T_\theta(\delta_k)$ from a random initialisation on 24 real AirfRANS cases,
the measured per-step ratio $\|\delta_{k+1}-\delta^{*}\|/\|\delta_k-\delta^{*}\|$ (in the
geometric regime, before the $\sim 10^{-7}$ solve floor) has \textbf{median 0.78} and
maximum $0.86$---strictly below the design bound $\kappa = 0.9$ and the falsification
threshold $1$---and reaches a relative distance of $10^{-5}$ to the fixed point in a
median of 37 steps (\texttt{results/certificates/h5\_contraction.json}). The corrector
contracts as designed; this is a property of the \emph{operator}, separate from whether
minimising the output's PDE residual helps the field (it does not, \autoref{sec:iters}).

\paragraph{In-distribution coverage (H4).} With per-cell MC-dropout $\sigma$ (16 passes)
and a 100/100 calibration/test split of the AirfRANS test set, split-conformal
calibration at $\alpha = 0.1$ attains per-channel coverage of \textbf{0.911 ($u$), 0.928
($v$), 0.942 ($p$)}---all in the $[0.85, 0.95]$ band and conservatively above the $0.90$
target, as the $\ge 1-\alpha$ guarantee requires. The conformal multipliers ($q = 0.77,
1.30, 1.75$) both tighten an over-dispersed and widen an under-dispersed raw $\sigma$, so
the calibration is doing real work. Coverage tracks $1-\alpha$ across a sweep of $\alpha$
(\autoref{fig:fig6}). The reliability diagram is exact at the target level but over-covers
at lower nominal levels (ECE $u/v/p = 0.064/0.093/0.130$), the signature of a heavy-tailed
$|\text{error}|/\sigma$ ratio; we therefore claim coverage at the chosen $\alpha$,
\textbf{not} full distributional calibration
(\texttt{results/certificates/h4\_coverage.json}).

\paragraph{The certificate is calibrated on the deployed, corrected field---not the raw
backbone.} NeuroForge ships the DEQ-corrected field, so the coverage guarantee must hold
for \emph{that} output; split-conformal calibration, however, applies to whatever
predictor it scores, and the raw-backbone certificate does \textbf{not} transfer unchanged
through the corrector. We verify this on the dropout-FNO backbone whose DEQ corrector
materially changes the field (mean $|\text{corrected}-\text{raw}|\approx 5$), through the
\emph{same} \texttt{cases\_to\_error\_sigma} path as H4 but with a corrected predict
function (\texttt{scripts/verify\_conformal\_corrected\_field.py}; $15/15$
calibration/test). The clean isolation is a \emph{within-run frozen-$q$ contrast}---same
fitted $q$, same input-conditional $\sigma$, same cases, only the raw vs.\ corrected error
map---and it shows the direction of the coverage change tracking the direction of the error
change: where the corrector \emph{enlarges} the error the raw band drops below the $0.90$
target ($v$ $0.896\rightarrow0.819$, $p$ $0.864\rightarrow0.777$), and where it
\emph{shrinks} the error the raw band over-covers and stays conservative ($u$
$0.856\rightarrow0.888$). This is the coverage analogue of detector $\ne$ fixer
(\autoref{sec:iters}): the DEQ operator contracts in its \emph{own} variable $\delta$ (H5,
median ratio $0.78$), but the Banach bound governs the correction, not $|\text{error}|/
\sigma$, so it cannot imply coverage survival---and on $v/p$ residual contraction coincides
with error \emph{growth}, which is exactly what breaks the raw band. Re-fitting $q$ on the
corrected calibration field---at no change to $\sigma$, which is input-conditional and
identical for both fields---restores coverage toward target where it had collapsed ($v$
$0.819\rightarrow0.894$, $p$ $0.777\rightarrow0.834$) and sharpens it where it was loose
($u$ $0.888\rightarrow0.874$). We therefore calibrate the shipped certificate on the
corrected output; split-conformal validity requires only exchangeable corrected-field
scores, not that $\sigma$ be the corrected field's own uncertainty. These estimates are
directional at $n=15$ (per-case coverage std $\approx 0.15$--$0.23$, so $\approx\pm0.05$
standard error), but the same pattern---raw-$q$ under-covers $v/p$ on the corrected field,
refitting restores it---reproduces on an independent $20/20$ split
(\texttt{results/certificates/probe\_conformal\_after\_deq.json}); the residual $p$
under-coverage is the small-sample heavy-tail effect already noted at $n=100$. On the
deployed SOTA Transolver backbone we now verify this \emph{directly} rather than by
extrapolation: through the same conformal path, with $\sigma$ the frozen five-member
deep-ensemble std and the scored field replaced by the DEQ-corrected output, at the
production $100/100$ split and across all three corrector seeds
(\texttt{scripts/run\_w2\_conformal\_corrected.py},
\texttt{results/uq\_ensemble/w2\_conformal\_corrected.json}; a faithfulness gate reproduces
the published backbone $p$-certificate---$q{=}2.35$, coverage $0.915$---exactly). Re-fitting
$q$ on the corrected calibration field gives a valid certificate (coverage $u/v/p =
0.88/0.91/0.91$ vs the $0.90$ target), and because the corrected field's error is lower at
frozen $\sigma$, that accuracy propagates into the certificate as proportionally tighter
intervals at the same coverage ($q$: $v$ $2.58\rightarrow2.34$, $p$ $2.35\rightarrow2.10$ on
every seed). Applying the \emph{backbone's} own $q$ to the corrected field---the direct
``remains-conservative'' test---over-covers on $v$ and $p$ ($\approx 0.93$), as the field-error
reduction (\autoref{sec:v2}) predicts, while $u$ holds at $\approx 0.88$, just below target:
this tracks the backbone's own heavy-tailed $u$ band ($0.881$) and is inherited from it, not
introduced by the corrector. We therefore ship the certificate refit on the corrected output;
recalibration matters most for correctors that trade field error for residual reduction (the
dropout-FNO above), and is benign-to-beneficial for the deployed Transolver, whose bands it
sharpens.

\paragraph{Split-resampling robustness (20 splits).} Because the certificates above are
computed on one $100/100$ calibration/test draw, we re-draw the split $20$ times
(\texttt{scripts/run\_multisplit\_conformal.py},
\texttt{results/uq\_ensemble/multisplit\_conformal.json}; same code path, frozen ensemble
$\sigma$, $q$ refit per split). Coverage is tightly centred on the $0.90$ target in
\emph{every} arm and channel ($0.895$--$0.902$, split-to-split std $\le 0.014$); the
published seed-$0$ pressure coverage of $0.9155$ sits at the upper end of its distribution
($0.899 \pm 0.009$). The corrected-field sharpening is robust, not a split artifact: on
pressure, $q$ drops from $2.26 \pm 0.05$ (backbone) to $2.02$--$2.07$ for all three
corrector seeds, and on $v$ from $2.55 \pm 0.05$ to $2.28$--$2.31$; on the heavy-tailed $u$
channel the effect is mixed (one seed sharpens, two widen slightly), consistent with the
$u$-neutral verdict above.

\paragraph{Out-of-distribution coverage (empirical, not a guarantee).} Calibrating $q$ on
the in-distribution \texttt{full} split and evaluating on the OOD splits, coverage stays
in the target band: $u$ $0.90/0.91/0.91$, $v$ $0.91/0.94/0.92$, $p$ $0.92/0.94/0.90$
(full / reynolds / aoa) vs the $0.90$ target
(\texttt{results/sensitivity/ood\_coverage.json}, \autoref{fig:fig5}). \textbf{This is an
empirical observation, not a guarantee.} The conformal guarantee assumes exchangeable
calibration and test draws, which fails under shift; coverage nonetheless holds because
the MC-dropout $\sigma$ inflates appropriately on shifted inputs, so the same $q$ still
brackets the (larger) errors. We report this as a desirable empirical property and
explicitly \textbf{do not} assert distribution-free coverage out-of-distribution.

\paragraph{From correlation to decisions: selective prediction.} A rank correlation of
$0.6$ can sound modest, so we also measure what the trust signal is \emph{for}: deciding
which predictions to reject (\texttt{scripts/run\_selective\_prediction.py},
\texttt{results/selective/selective\_prediction.json}; \autoref{fig:risk_coverage}). On the
$200$ AirfRANS test cases (ensemble-mean arm), the residual detects worst-decile-error
cases with AUROC $0.871$; the ensemble $\sigma$ alone gives $0.894$; and the rank-fused
residual$+\sigma$ score reaches \textbf{AUROC $0.905$} (AUPRC $0.681$ against a $0.10$
prevalence baseline). The same residual score transfers: AUROC $0.89$--$0.91$ on the three
DEQ-corrected arms and $0.83$--$0.98$ on the three DeepCFD OOD seeds. In risk--coverage
terms, rejecting the $10\%$ least-trusted cases by the fused score reduces the retained
mean relative-$L_2$ error from $0.0040$ to $0.0035$---recovering $\approx 91\%$ of the
error reduction an \emph{oracle} that rejects by true error could achieve at that budget
(the residual alone recovers $\approx 67\%$; random rejection recovers none). Case-level
bootstrap $95\%$ CIs ($10^4$ resamples;
\texttt{results/control/bootstrap\_spearman\_ci.json}) put the ensemble-arm residual
correlation at $0.610$ $[0.506, 0.698]$ and the fused score at $0.703$ $[0.615, 0.773]$;
every arm's CI excludes zero by a wide margin. The trust signal is thus not merely
correlated with error; it supports near-oracle triage at deployment-relevant rejection
budgets. The audit is also essentially free at the point of use: computing the full
residual diagnostic and trust map, the case-level score, the ensemble spread, and the
conformal band costs a median $1.7$\,ms per case on a single CPU thread
(\texttt{scripts/measure\_audit\_cost.py}, \texttt{results/control/audit\_cost.json}),
against the ${\sim}25$-minute, $16$-core RANS solve that produced each AirfRANS case
\citep{bonnet2022airfrans} --- the certificate costs ${\sim}10^{-6}$ of the computation
it vouches for.

\begin{figure}[t]
\centering
\includegraphics[width=0.95\textwidth]{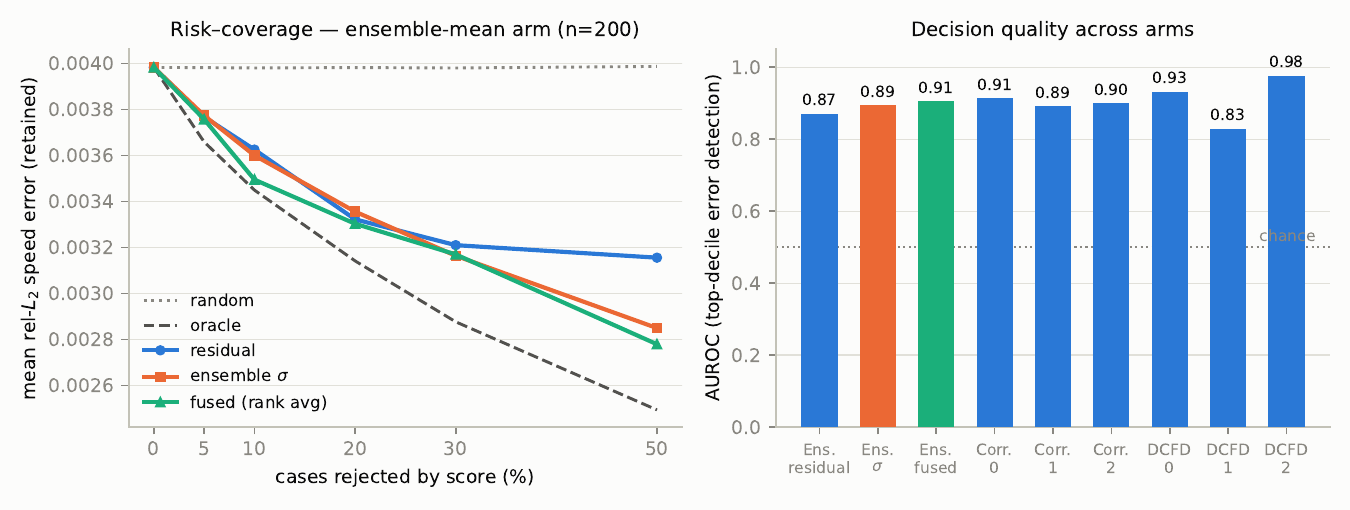}
\caption{\textbf{Selective prediction with the trust signal.} \emph{Left:} risk--coverage
curves on the AirfRANS test set (ensemble-mean arm): mean retained relative-$L_2$ error
after rejecting the least-trusted fraction of cases by physics residual, ensemble
$\sigma$, or their rank fusion, against the oracle (reject by true error) and random
rejection. \emph{Right:} AUROC for detecting worst-decile-error cases across arms
(AirfRANS ensemble-mean, three corrected seeds, three DeepCFD OOD seeds); chance $=0.5$.}
\label{fig:risk_coverage}
\end{figure}

The trust layer depends only on a model's predictions and $\sigma$, so it is
backbone-robust; \autoref{sec:v2} demonstrates it on a SOTA Transolver backbone, where
coverage holds at $0.902{\pm}0.008$ (target $0.90$) and a deep-ensemble $\sigma$ makes the
band input-adaptive on that near-deterministic model ($q=2.35$, ECE $0.074$;
\autoref{tab:uq}).

\subsection{Computational cost and complexity}
\label{sec:cost}

A trust layer is only deployable if auditing costs materially less than the prediction it
audits. We therefore time every stage of the deployed pipeline on the AirfRANS \texttt{full}
test set at resolution $128$ (\texttt{scripts/measure\_inference\_cost.py},
\texttt{results/control/inference\_cost.json}; $100$ cases, median of two timed repeats after
an untimed warm-up, CUDA synchronised around each repeat, on one RTX 4070 Ti with the
package's single-thread BLAS cap in force). \autoref{tab:cost} reports the result.

\begin{table}[t]
\centering
\small
\caption{\textbf{Per-case wall-clock cost of the deployed pipeline}
($128^2$ grid, $100$ AirfRANS \texttt{full} test cases, medians). \emph{Where} is the device
each stage actually runs on. The audit is the entire certificate: residual maps, trust map,
case-level score and conformal band. The classical anchor is the cost of generating one
AirfRANS case, as reported by its authors ($\approx 25$ min on $16$ CPU cores); it is
quoted for scale, not as a matched-hardware benchmark.}
\label{tab:cost}
\begin{adjustbox}{max width=\textwidth}
\begin{tabular}{l l r r l}
\toprule
stage & where & median & p95 & share of a deployed solve \\
\midrule
encode (geometry $\to$ input planes) & CPU, 1 thread & $1413$\,ms & $1620$\,ms & preprocessing \\
backbone forward + rasterise & GPU & $2080$\,ms & $3537$\,ms & $54.4\%$ \\
correction (Neural Residual Iteration) & GPU & $1747$\,ms & $2950$\,ms & $45.7\%$ \\
\midrule
\textbf{deployed solve} (backbone + correction) & GPU & $\mathbf{3822}$\,\textbf{ms} & $6400$\,ms & $100\%$ \\
\midrule
\textbf{audit} (residuals + trust + score + band) & CPU, 1 thread & $\mathbf{1.13}$\,\textbf{ms} & $2.49$\,ms & $\mathbf{0.03\%}$ \\
deep ensemble, $M{=}5$ (for the fused score) & GPU & $10\,114$\,ms & $10\,546$\,ms & $265\%$ \\
\midrule
classical steady-RANS solve \citep{bonnet2022airfrans} & 16 CPU cores & ${\sim}1\,500$\,s & --- & ${\sim}39\,000\%$ \\
\bottomrule
\end{tabular}
\end{adjustbox}
\end{table}

\paragraph{The certificate is free; the ensemble is not.} The audit costs $1.13$\,ms against
a $3822$\,ms solve---one part in $3\,400$, or $0.03\%$. End to end, a case costs $5.24$\,s
(encode $+$ solve $+$ audit) against the ${\sim}1500$\,s that produced its label, a
${\sim}286\times$ reduction, and the certificate is invisible inside that. This is the
quantitative form of the paper's claim that the surrogate can audit itself: the trust score,
the trust map and the conformal band are obtained for three hundredths of a percent of the
computation they vouch for, and they require no ground truth.

The deep-ensemble spread does not share that property. Measured directly by loading five
independent backbones, $M{=}5$ costs $4.86\times$ a single backbone---linear in $M$, as
expected, and confirmed rather than assumed---which is $2.1$ deployed solves of \emph{extra}
compute. The two halves of the fused score therefore sit at opposite ends of the cost scale,
and this sharpens the selective-prediction result of \autoref{sec:conformal} into a
deployment choice rather than a single operating point: the residual alone recovers
$\approx 67\%$ of the oracle's achievable error reduction at essentially zero marginal cost,
while fusing it with the ensemble spread reaches $\approx 91\%$ but triples the compute
budget per case. A practitioner who cannot afford five backbones still gets two thirds of the
benefit for nothing; one who can, buys the last third at a known, measured price.

\paragraph{Scaling.} The audit's operation count is $O(N)$ in the number of grid cells
$N=n_yn_x$: every term is a fixed-width finite-difference stencil, a mask, or a reduction.
Empirically, over a $16\times$ range in $N$ ($64^2$ to $256^2$, fields resampled purely to
vary array extent) the audit rises only $6.3\times$, a log--log slope of $0.64$
($R^2=0.99$)---sub-linear because per-call fixed costs (allocation, Python dispatch) still
dominate at these sizes. The asymptotic rate is the $O(N)$ of the operation count; the
practical consequence is that the audit's share of the pipeline \emph{falls} as grids grow,
since the backbone and the corrector both carry super-linear cost in resolution. Across the
test fleet the backbone's cost showed no measurable dependence on the native cloud size, but
that range is only $1.21\times$ wide ($163\,590$--$197\,146$ points) and cloud size explains
$3\%$ of the timing variance ($R^2=0.03$); we therefore report the observation and decline to
fit a scaling law to it.

\paragraph{What these numbers are not.} The backbone stage times the Transolver forward pass
\emph{and} the rasterisation of its point-cloud output onto the grid; the latter is
single-threaded NumPy under the package's BLAS cap, so the GPU figures are conservative upper
bounds on the model itself rather than a tuned inference benchmark. The encode stage is
likewise unoptimised CPU preprocessing and would be cached in any real design loop, where the
same geometry is queried at many operating points. The classical anchor comes from the
AirfRANS authors on $16$ CPU cores \citep{bonnet2022airfrans}; ours is one consumer GPU plus
one CPU thread, so the ${\sim}286\times$ figure indicates the order of the gap, not a
controlled speed-up measurement. None of these caveats touch the ratio the paper actually
depends on---audit against solve---because both sides of it were measured in the same run, on
the same cases, under the same thread caps.

\subsection{Honest baseline: the grid backbone trails SOTA}
\label{sec:baseline}

\begin{table}[t]
\centering
\small
\caption{Matched-budget baseline on AirfRANS \texttt{full} ($\texttt{n\_train}=800$, 80
epochs, identical rasterisation and \texttt{evaluate\_cases} scoring), 3 seeds. See
\autoref{fig:fig3}.}
\label{tab:transolver}
\resizebox{\textwidth}{!}{%
\begin{tabular}{l r r r r r r r}
\toprule
model & \texttt{n\_params} & $\texttt{mse\_u}$ & $\texttt{mse\_v}$ & $\texttt{mse\_p}$ & surf.\ $\texttt{mse\_p}$ & $\rho_{C_l}$ & $\rho_{C_d}$ \\
\midrule
Transolver (baseline) & 7.35M & $0.120{\pm}0.005$ & $0.088{\pm}0.013$ & $628.5{\pm}29$ & $9110{\pm}500$ & $0.9992{\pm}0.0002$ & $0.9963{\pm}0.0012$ \\
our backbone (\autoref{tab:indist}) & --- & $3.479$ & $0.385$ & $2444.8$ & $548823$ & $0.987$ & $0.895$ \\
our backbone + DEQ & --- & $3.457$ & $0.880$ & $3832.7$ & $361681$ & $0.986$ & $0.923$ \\
\bottomrule
\end{tabular}%
}
\end{table}

Transolver, a SOTA point-cloud physics-attention transformer, trained on the same data
with the same budget and scored identically, is far ahead of our grid backbone: relative
to the bare backbone, roughly \textbf{$29\times$ on $\texttt{mse\_u}$}, \textbf{$4\times$
on $\texttt{mse\_v}$} ($\approx 10\times$ against the volume-regressed DEQ arm),
\textbf{$4\times$ on $\texttt{mse\_p}$}, and \textbf{$\sim 60\times$ on surface pressure}
($\approx 40\times$ against the DEQ arm, whose surface MSE is lower), with near-perfect
force ranking. Our grid backbone is not a competitive surrogate; we make \textbf{no
competitiveness claim} for it. The decisive next step follows in \autoref{sec:v2}: we put
the trust layer and the DEQ corrector \emph{on} this Transolver backbone, where the residual
is a calibrated detector ($\rho=0.611$) and the corrector reduces field MSE by $8$--$25\%$.
The large grid-backbone gap is itself informative: even a much weaker backbone yields a
residual that is a calibrated, shift-robust detector of its own errors.

\paragraph{Note on the two Transolver runs.} The Transolver numbers here
(\autoref{tab:transolver}: $\texttt{mse\_u}$ $0.120{\pm}0.005$, $\texttt{mse\_v}$
$0.088{\pm}0.013$, $\texttt{mse\_p}$ $628.5{\pm}29$) and the v2 backbone-alone numbers
(\autoref{tab:v2}, 5 seeds: $0.133{\pm}0.017$, $0.096{\pm}0.006$, $644{\pm}63$) are the
\emph{same Transolver architecture at the same budget} ($7.35$M params, 80 epochs), trained
in two separate runs; the differences are within the reported per-seed spread. The v2 run is
the one we use as the NeuroForge backbone in \autoref{sec:v2}.

\paragraph{Positioning against the published AirfRANS field.} Our trust layer and corrector
run on a Transolver backbone---a published SOTA architecture \citep{wu2024transolver}---so
competitiveness comes from that architectural choice, not from re-benchmarking our own metric
values. We deliberately do \emph{not} place a NeuroForge row in a cross-method comparison:
AirfRANS results are reported in mutually-incompatible conventions (the AirfRANS-paper baselines
use standardized-field MSE and force coefficients against official OpenFOAM labels;
Transolver-lineage papers report relative-$L_2$; our pipeline uses physical-unit MSE and
self-integrated forces, \autoref{sec:v2}), so any shared cell would mislead. \autoref{tab:airfrans-sota}
reproduces the canonical AirfRANS-paper baselines as field context. Note the benchmark's
signature difficulty---drag ranking $\rho_{C_d}$ is near zero or negative for every baseline,
the headline finding of \citet{bonnet2022airfrans}---which is also why we read our own
$\rho_{C_d}$ as internal self-consistency rather than a solved drag-ranking claim.

\begin{table}[t]
\centering\small
\caption{Published AirfRANS \texttt{full}-task baselines from \citet{bonnet2022airfrans}, as
field context only. Volume MSE on standardized fields ($\times10^{-2}$); surface $\bar p$
($\times10^{-1}$); force relative errors (\%); $\rho$ = Spearman rank correlation vs.\
\emph{official} force labels. NeuroForge is \emph{not} a comparable row (physical-unit MSE,
self-integrated forces; \autoref{sec:v2}). The near-zero/negative $\rho_{C_d}$ is the
benchmark's known drag-ranking difficulty.}
\label{tab:airfrans-sota}
\begin{tabular}{l r r r r r r}
\toprule
method & $\bar u_x$ & $\bar p$ & surf.\ $\bar p$ & $C_L$ err & $C_D$ err & $\rho_{C_l}\,/\,\rho_{C_d}$ \\
\midrule
MLP         & $0.95$ & $0.74$ & $1.13$ & $0.77$  & $4.29$  & $0.913\,/\,{-}0.117$ \\
GraphSAGE   & $0.83$ & $0.66$ & $0.66$ & $0.52$  & $4.05$  & $0.965\,/\,{-}0.303$ \\
PointNet    & $3.50$ & $1.15$ & $0.93$ & $0.74$  & $14.64$ & $0.938\,/\,{-}0.022$ \\
Graph U-Net & $1.52$ & $0.66$ & $0.39$ & $0.49$  & $10.39$ & $0.967\,/\,{-}0.138$ \\
\bottomrule
\end{tabular}
\end{table}

\subsection{Figures}

\begin{figure}[t]
\centering
\includegraphics[width=0.85\textwidth]{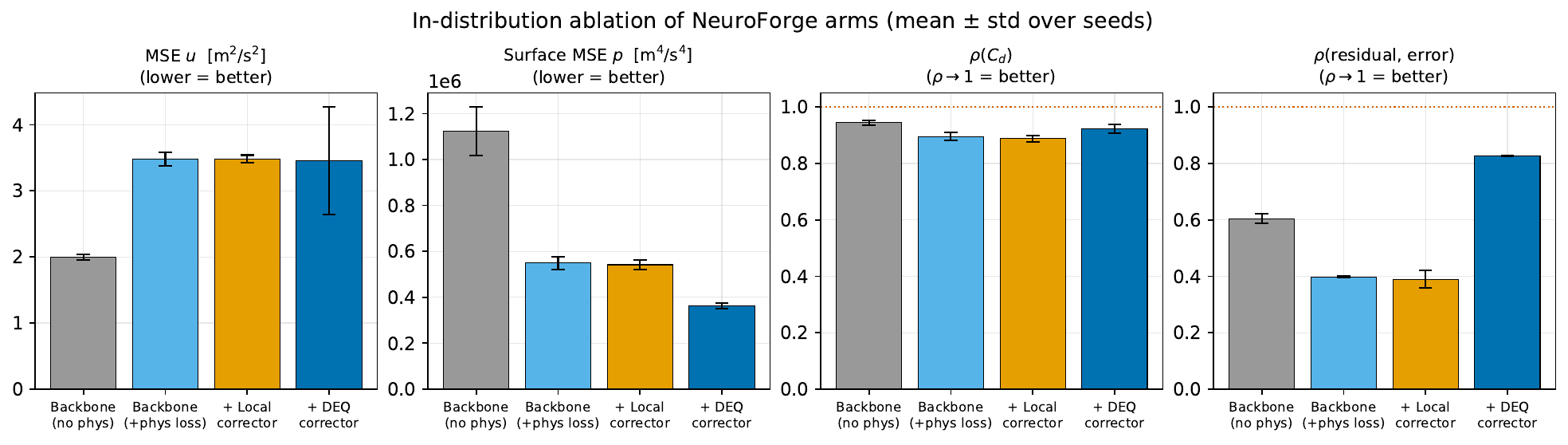}
\caption{In-distribution ablation (\autoref{tab:indist}): the DEQ corrector lifts the
trust signal ($0.40\rightarrow0.83$) and surface fidelity while regressing volume
$\texttt{mse\_v}$/$\texttt{mse\_p}$.}
\label{fig:fig1}
\end{figure}

\begin{figure}[t]
\centering
\includegraphics[width=0.85\textwidth]{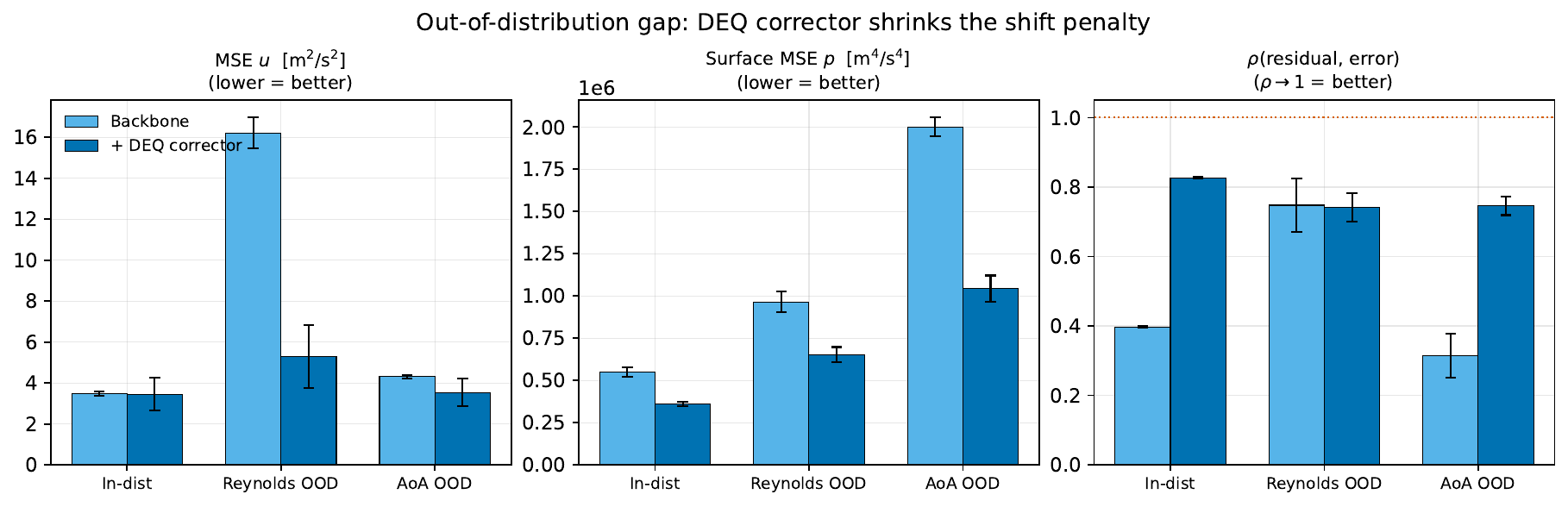}
\caption{OOD gap and the regime-invariant trust signal (\texttt{aoa}
$0.31\rightarrow0.75$); the bare backbone's signal collapses, the corrected model's holds.}
\label{fig:fig2}
\end{figure}

\begin{figure}[t]
\centering
\includegraphics[width=0.85\textwidth]{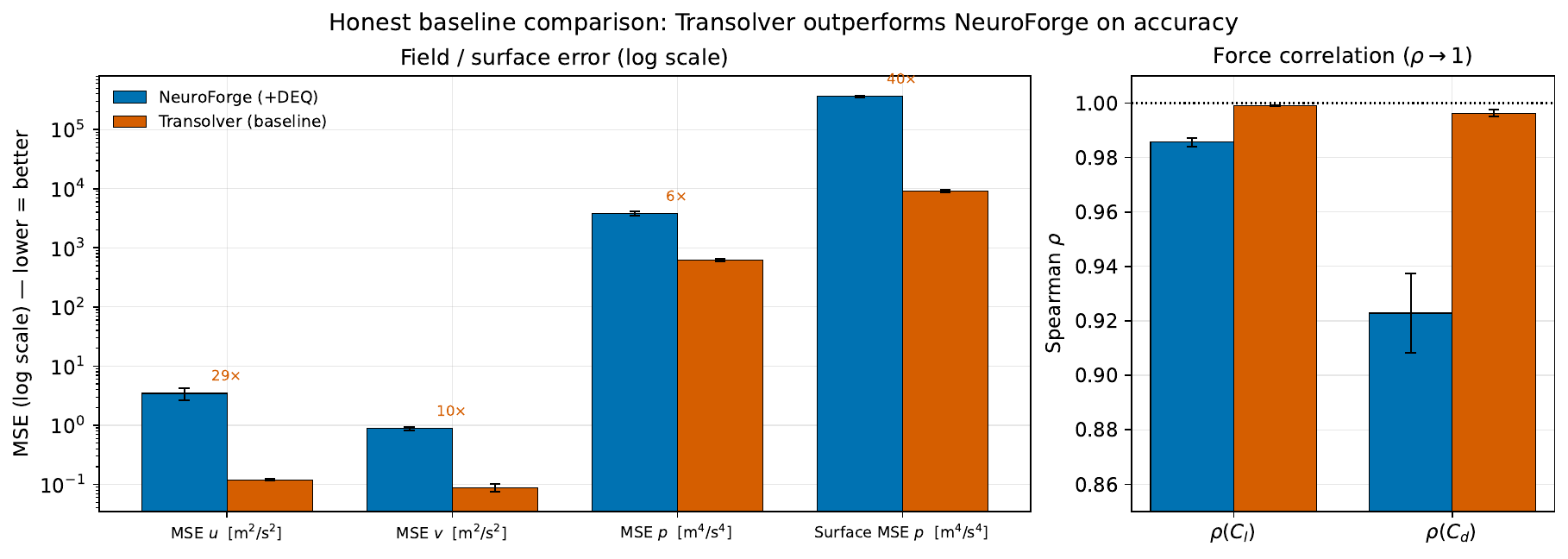}
\caption{Honest matched-budget baseline; the grid backbone trails Transolver
$\sim 4$--$60\times$ across channels (no competitiveness claim).}
\label{fig:fig3}
\end{figure}

\begin{figure}[t]
\centering
\includegraphics[width=0.85\textwidth]{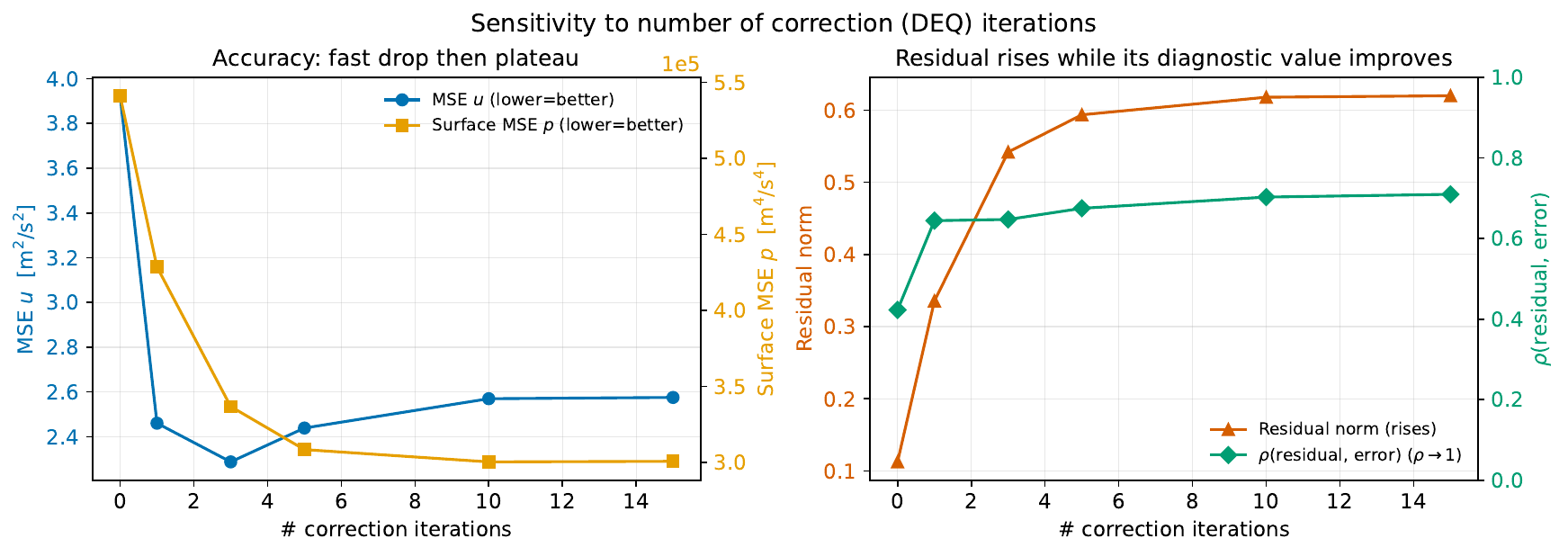}
\caption{The money figure for the headline: PDE residual \emph{rises} while field error
\emph{falls} across correction iterations (detector $\ne$ fixer).}
\label{fig:fig4}
\end{figure}

\begin{figure}[t]
\centering
\includegraphics[width=0.85\textwidth]{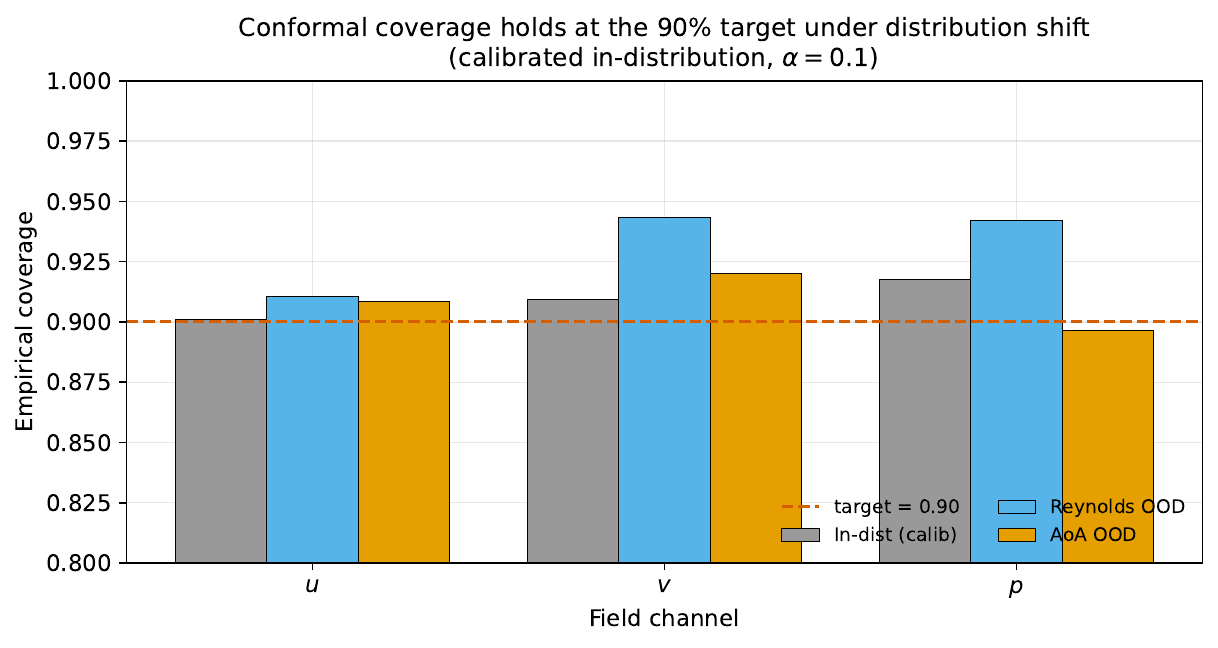}
\caption{Conformal coverage stays in-band on OOD splits (empirical, via $\sigma$
inflation).}
\label{fig:fig5}
\end{figure}

\begin{figure}[t]
\centering
\includegraphics[width=0.85\textwidth]{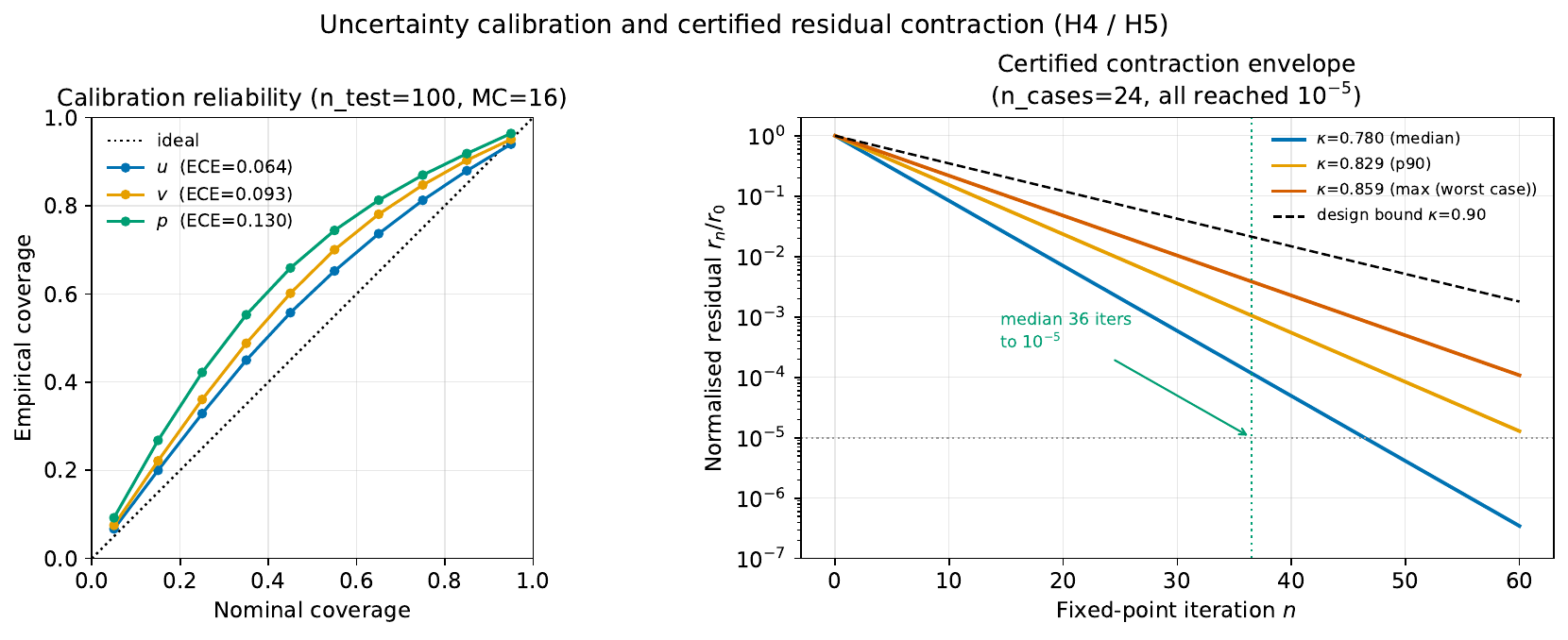}
\caption{In-distribution reliability / coverage-vs-$\alpha$ and the measured DEQ
contraction ($0.78 < 1$).}
\label{fig:fig6}
\end{figure}

\begin{figure}[t]
\centering
\includegraphics[width=0.75\textwidth]{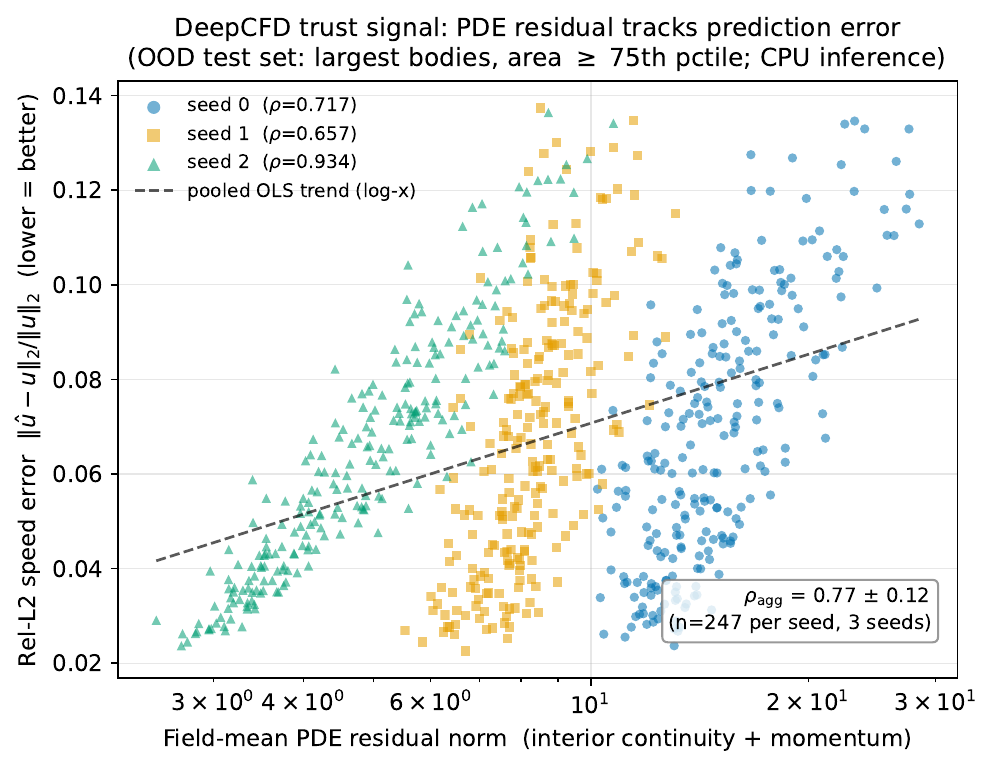}
\caption{The trust signal generalizes to a second dataset (DeepCFD, 2-D laminar flow over bluff
bodies; OOD held-out-largest-body split, 3 seeds). Per-case field-mean PDE residual norm
(interior continuity+momentum) vs.\ rel-$L_2$ speed error: a low residual tracks a low error, as
on AirfRANS. Per-case Spearman $\rho=0.77{\pm}0.12$ (per seed $0.72/0.66/0.93$; the trend holds
\emph{within} each seed---the colour-coded clusters span different residual ranges by
initialisation). The residual is a reference-free error proxy on a geometry family and flow
regime absent from AirfRANS.}
\label{fig:deepcfd}
\end{figure}

\begin{figure}[t]
\centering
\includegraphics[width=\textwidth]{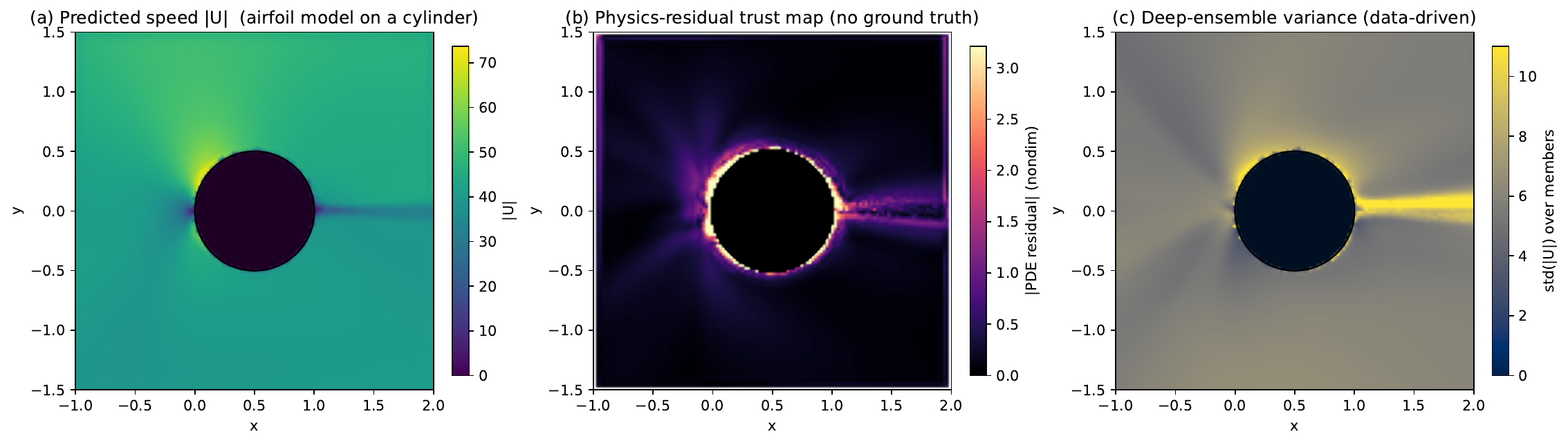}
\caption{Cross-geometry OOD (qualitative): an airfoil-trained model on a cylinder.
(a)~predicted speed (visibly broken---the flow is not arrested at the body); (b)~the
physics-residual map (no ground truth used), elevated field-wide; its bright near-body ring is
a generic no-slip-wall feature common to all bodies (a matched in-distribution-airfoil control
shows the same near/far ratio), \emph{not} an OOD localiser; (c)~deep-ensemble variance (the
uncorrected five members, the source of $\sigma$), for comparison. The OOD signal is the
field-wide residual \emph{magnitude}---$\approx 7\times$ higher than in-distribution airfoils
in the far field, away from any wall (\autoref{sec:ood})---not its spatial concentration.}
\label{fig:cylinder}
\end{figure}

\subsection{Reading against the pre-registered hypotheses}

\begin{itemize}
\item \textbf{H1 (the corrector improves accuracy).} \emph{Backbone-dependent.} On the
\textbf{weak grid backbone} it is \emph{not supported as pre-registered}: the feed-forward
corrector helps on nothing and the DEQ corrector regresses volume
$\texttt{mse\_v}$/$\texttt{mse\_p}$ (3/3) while flat on $\texttt{mse\_u}$. On the
\textbf{SOTA Transolver backbone} it \emph{is supported}: the DEQ corrector reduces all
three volume MSE channels on all 5 seeds ($\texttt{mse\_u}$ $-8\%$, $\texttt{mse\_v}$
$-21\%$, $\texttt{mse\_p}$ $-25\%$; \autoref{sec:v2}, \autoref{tab:v2}) with force ranking
held. Correction quality is therefore backbone-dependent---a finding, not a hidden
negative---and the supervised corrector works where the backbone is strong.
\item \textbf{H2 (the residual is a valid trust signal).} \emph{Supported, robustly, in-
and out-of-distribution.} resid$\leftrightarrow$err $\rho > 0$ for every arm/seed/split;
DEQ lifts it to $0.83$ in-dist and makes it regime-invariant ($\approx 0.74$) OOD,
rescuing the \texttt{aoa} collapse $0.31\rightarrow0.75$.
\item \textbf{H3 (DEQ $\ge$ feed-forward corrector).} \emph{Holds on the design axis
only.} DEQ beats the feed-forward corrector on surface MSE, $\rho_{C_d}$, and the trust
signal (3/3), but is worse on volume $\texttt{mse\_v}$/$\texttt{mse\_p}$; the two trade
off.
\item \textbf{H4 (conformal coverage).} \emph{Supported in-distribution}
($0.91/0.93/0.94$ at $\alpha=0.1$; SOTA backbone $0.902$, deep-ensemble band adaptive at
ECE $0.074$, \autoref{tab:uq}); \emph{empirical only OOD} (\autoref{sec:conformal}).
\item \textbf{H5 (DEQ contraction).} \emph{Supported} (measured factor
$0.78 < \kappa = 0.9 < 1$).
\end{itemize}

Every number comes from the committed, reproducible harness
(\texttt{benchmarks/ablation.py}, \texttt{scripts/run\_certificates.py},
\texttt{results/sensitivity/}).


\section{The residual floor: a self-consistent monitor is a good detector but a bad fixer}
\label{sec:residual_floor}

The iteration sweep of \autoref{sec:iters} (\autoref{tab:iters}) shows a counter-intuitive
fact: driving the monitored physics residual \emph{down} drives the field error \emph{up}.
We now explain this exactly. The mechanism is not ill-conditioning per se; it is that the
\emph{monitored} residual is self-consistent at a field that is not the truth---so its
minimiser sits away from $u^\star$ by an irreducible amount. The same decomposition that
forces this also explains why the residual is nonetheless an excellent error
\emph{detector}: one equation, both verdicts.

\subsection{Setup: the monitored operator}
\label{sec:rf_operator}

Let $R_h$ be the residual operator the framework actually monitors and (in the
objective/gate role) minimises: the \emph{discrete} steady incompressible RANS
\textbf{continuity and momentum} residual on the $(n_y,n_x)$ grid, evaluated on physical
fields with $\nu_{\mathrm{eff}}=\nu+\nu_t$. By construction (\texttt{Diagnostics.residual\_norm};
\autoref{sec:method}) it \textbf{omits the no-slip wall term}. Write the objective
$J(u)=\tfrac12\|R_h(u)\|^2$. Let $u^\star$ be the discrete ground-truth field, define the
\emph{residual floor} $r^\star:=R_h(u^\star)$, and for any prediction $\hat u$ write the
error $e:=\hat u-u^\star$ with the first-order expansion
\begin{equation}
R_h(\hat u)=r^\star+L\,e+o(\|e\|),\qquad L:=DR_h(u^\star).
\label{eq:rf-decomp}
\end{equation}

\begin{theorem}[Residual floor: misplaced minimum, quantified floor, and local divergence]
\label{thm:residual-floor}
Assume \textup{(H1)} $R_h$ is the monitored discrete operator above (continuity$+$momentum,
no-slip omitted); and \textup{(H2)} the floor is nonzero, $r^\star:=R_h(u^\star)\neq 0$. Write
$L:=DR_h(u^\star)$ and let $\sigma_{\min}$ denote its smallest \emph{nonzero} singular value
(no injectivity is assumed; $L$ has a non-trivial kernel---see the remark below). Then:
\begin{enumerate}
\item[\textup{(i)}] \emph{(Misplaced minimum, exact.)} The uniform-freestream field $u_\infty$
satisfies $R_h(u_\infty)=0$ exactly, so $\min_u J=0$; whereas $J(u^\star)=\tfrac12\|r^\star\|^2>0$.
Hence $u^\star$ is \textbf{not} a minimiser of $J$.
\item[\textup{(ii)}] \emph{(Misplaced stationarity.)} $\nabla J(u^\star)=L^\top r^\star$, so
$u^\star$ is a stationary point of $J$ iff $L^\top r^\star=0$; if $r^\star\notin\ker L^\top$,
gradient flow on $J$ leaves $u^\star$.
\item[\textup{(iii)}] \emph{(Quantified floor.)} The \emph{minimum-norm} minimiser of the
linearised objective $\tfrac12\|r^\star+Le\|^2$ is $e_\infty=-L^{+}r^\star$ ($L^{+}$ the
Moore--Penrose pseudoinverse), leaving irreducible field error
$\|e_\infty\|=\|L^{+}r^\star\|\le\|r^\star\|/\sigma_{\min}$, which grows as
$\sigma_{\min}\to 0$ (ill-conditioning).
\item[\textup{(iv)}] \emph{(Local divergence direction.)} Under residual-gradient flow
$\dot e=-\nabla J(\hat u)$, the field error obeys
$\frac{d}{dt}\tfrac12\|e\|^2\big|_{0}=-\|Le\|^2-(Le)\cdot r^\star$. There is an open cone of
error directions---those with $(Le)\cdot r^\star<-\|Le\|^2$---on which a residual-reducing
step strictly \textbf{increases} the field error.
\end{enumerate}
\end{theorem}

\begin{proof}
(i) For a spatially constant field every finite-difference derivative vanishes, so the
discrete continuity and momentum residuals are identically zero: $R_h(u_\infty)=0$, giving
$J(u_\infty)=0=\min J$ (as $J\ge0$). By (H2), $J(u^\star)=\tfrac12\|r^\star\|^2>0$.
(ii) $\nabla J(u)=DR_h(u)^\top R_h(u)$; evaluate at $u^\star$ and use \eqref{eq:rf-decomp}.
(iii) $\tfrac12\|r^\star+Le\|^2$ is convex in $e$; its minimum-norm minimiser is the
least-squares solution $e_\infty=-L^{+}r^\star$, which requires no injectivity (it is the
Moore--Penrose pseudoinverse of the normal equations $L^\top L e=-L^\top r^\star$), and
$\|L^{+}\|_2=1/\sigma_{\min}$ with $\sigma_{\min}$ the smallest nonzero singular value of $L$.
(iv) To first order $\dot e|_0=-L^\top(r^\star+Le)$, so
$\frac{d}{dt}\tfrac12\|e\|^2|_0=e\cdot\dot e=-e^\top L^\top(r^\star+Le)=-\|Le\|^2-(Le)\cdot r^\star$.
This is positive iff $(Le)\cdot r^\star<-\|Le\|^2$; taking $Le$ near $-\alpha\,r^\star/\|r^\star\|$
(feasible whenever $r^\star$ has a nonzero component in $\operatorname{range}L$) the
condition holds on an open cone with $0<\alpha<\|r^\star\|$.
\end{proof}

\paragraph{Good detector, bad fixer---from one decomposition.} Equation~\eqref{eq:rf-decomp}
splits the monitored residual into a fixed floor $r^\star$ and an error-driven term $Le$.
\emph{Far} from the truth, $\|Le\|\gg\|r^\star\|$, so $\|R_h(\hat u)\|\approx\|Le\|$, and since
$\sigma_{\min}(L)\|e\|\le\|Le\|\le\sigma_{\max}(L)\|e\|$ the residual norm is a two-sided proxy
for the field error (up to $\kappa(L)$)---this is the calibrated \emph{trust signal}
(rank correlation up to $0.83$, \autoref{sec:indist-ablation}). \emph{Near} the truth,
$\|Le\|\lesssim\|r^\star\|$, the floor dominates, and the residual decouples from the error:
its minimiser is displaced from $u^\star$ by \autoref{thm:residual-floor}(iii) and a
residual-reducing step can grow the error by (iv)---this is the failure of the residual as a
correction \emph{objective} (\autoref{sec:iters}). Note the divergence in (iv) vanishes when
$r^\star=0$: it is driven by the \emph{floor}, not by ill-conditioning, which merely sets the
floor's magnitude in (iii).

\paragraph{The kernel: modes the monitor can neither detect nor fix.} $L$ is genuinely
non-injective---$R_h$ maps the four field channels $(u,v,p,\nu_t)$ to three residual channels
(continuity, two momentum), and a uniform pressure offset is an exact null mode (the monitored
$R_h$ sees $p$ only through $\nabla p$, and omits the no-slip term that would otherwise pin the
gauge). Error components in $\ker L$---a constant pressure shift, and the residual-blind part of
$\nu_t$---are therefore \emph{invisible} to the monitor: it can neither flag them (detector) nor
remove them (fixer). The non-trivial kernel is not a defect of the bound but part of the
phenomenon: a self-consistent monitor is silent on its own null modes, which is why we report
coverage on the pressure channel up to a gauge and never claim the residual certifies $\nu_t$.

\paragraph{Empirical confirmation.} On $200/200$ real AirfRANS test cases, the monitored
residual of the ground-truth field is substantially nonzero (\textbf{(H2) holds}):
$\|r^\star\|$ has mean $0.192$ (median $0.133$), while the uniform-freestream field gives
$\|R_h(u_\infty)\|=0$ in every case---the objective strictly prefers a physically wrong field
to the truth (\texttt{results/certificates/residual\_floor\_realdata.json}). A trained
backbone's over-smoothed prediction (here the dropout-FNO of
\texttt{checkpoints/certificates\_deq.pt}, distinct from the Transolver headline) sits
\emph{below} the truth floor in $160/200$ cases, consistent with the loop's drive away from
$u^\star$ in \autoref{tab:iters}.

\paragraph{Robustness to the boundary term.} One might object that the monitored $R_h$
omits the no-slip term, so the negative is merely an artifact of dropping the boundary
constraint. We tested this with the \emph{framework's own} boundary-inclusive residual
$\rho_{\rm res}=\sqrt{r_c^2+r_x^2+r_y^2+r_{bc}^2}$ (the trust-map definition, $r_{bc}$ the
proximity-weighted no-slip penalty) along the correction sweep (dropout-FNO$+$DEQ, $n{=}10$,
same regime as \autoref{tab:iters}; \texttt{results/control/bc\_inclusive\_sweep.json}). The
negative \emph{survives}: as the corrector cuts field error ($\texttt{mse\_u}$
$3.17\!\to\!2.32$), the boundary-\emph{inclusive} residual rises ($0.145\!\to\!0.811$) in
lockstep with the boundary-excluded one ($0.125\!\to\!0.807$)---along the correction path the
prediction and the truth both approximately satisfy no-slip, so $r_{bc}$ stays small and
roughly constant and the interior floor dominates. (Empirically the uniform field also keeps a
\emph{lower} $\rho_{\rm res}$ than the truth, $0.073$ vs $0.187$, because $r_{bc}$ is
non-dimensionalised by $U$ and contributes only $\approx0.07$ in RMS.) The objection generalises---\emph{some} no-slip weight might rescue the residual---and
we can close it in full rather than concede it. Writing the weighted monitor
$\rho_\lambda=\sqrt{\|R_h\|^2+\lambda\,\overline{r_{bc}^2}}$, both logged components are
recoverable at every sweep point, so $\lambda$ can be swept post hoc at no computational
cost (\texttt{scripts/bc\_weight\_sweep.py}, \texttt{results/control/bc\_weight\_sweep.json}).
Both legs survive for \emph{every} $\lambda\ge0$, in closed form rather than on a grid of
sampled values. \textbf{(A)}~The uniform field remains a spurious minimum: it has both a
smaller interior residual and a smaller no-slip term than the truth
($\overline{r_{bc}^2}=0.0053$ versus $0.0073$---on a $128^2$ grid the boundary layer is
sub-cell, so the rasterised truth itself carries near-wall velocity), hence
$\rho_\lambda(u_\infty)<\rho_\lambda(u^\star)$ for all $\lambda$. \textbf{(B)}~Along the
correction path, where field error falls ($3.17\!\to\!2.32$), \emph{both} increments are
positive---the interior term by $+0.635$ and the no-slip term by $+0.0019$---so
$\rho_\lambda$ rises for all $\lambda$; there is no crossover weight. Up-weighting the
boundary condition cannot convert this residual into a usable correction objective,
because the no-slip violation grows along the very path that reduces the error.

\paragraph{Assumptions, stated plainly.} The theorem is about the \emph{monitored discrete}
operator $R_h$---what the loop minimises---not the continuous PDE; (H2) holds here because
$R_h$ omits the no-slip closure \emph{and} the $128^2$ grid under-resolves the boundary layer
and drops the spatially-varying-$\nu_t$ term $\nabla\nu_t\!\cdot\!\nabla u$. The floor is a
\emph{genuine} displacement whenever $r^\star$ has a component in $\operatorname{range}L$,
which it does---$r^\star$ is a discretisation/closure residual, not a pure gauge null mode.
The detector leg is a \emph{regime} statement: its tightness is controlled by the conditioning
of $L$ on $(\ker L)^\perp$, so we claim a positive rank correlation, not an isometry.

\paragraph{Relation to prior work.} \citet{krishnapriyan2021failure} tie ill-conditioning to
the \emph{optimisation difficulty} of PINN losses, and \citet{wang2022ntk} analyse spectral
bias and loss-term imbalance in PINN \emph{training dynamics}; both concern trainability, not
a displacement of the residual minimiser from the truth. The forward-vs-backward error gap
$\|e\|\!\le\!\|R_h\|/\sigma_{\min}$ is classical \citep{trefethenbau1997}, and a-posteriori
PINN bounds upper-bound error \emph{by} residual (the benign regime). That a residual blind to
certain modes makes the truth a non-minimiser is, in its bare form, elementary (the pressure
gauge above). Our contribution is sharper and threefold: \textbf{(a)} the \emph{quantified}
operator-specific floor $\|e_\infty\|=\|L^{+}r^\star\|\le\|r^\star\|/\sigma_{\min}$ for this
steady-RANS monitor; \textbf{(b)} the \emph{detector/fixer duality}---a single decomposition
\eqref{eq:rf-decomp} yielding both the positive trust-signal verdict and the negative
objective verdict, including the undetectable-and-uncorrectable kernel modes; and \textbf{(c)}
the demonstration on a learned neural-operator CFD corrector (AirfRANS), where this is exactly
the dissociation our experiments report.

\section{Limitations}
\label{sec:limitations}

\begin{itemize}
\item \textbf{Backbone-dependent correction quality.} The supervised DEQ corrector reduces
field MSE by $8$--$25\%$ on the SOTA Transolver backbone (\autoref{sec:v2}) but is
flat-to-unhelpful on the weak grid backbone (\autoref{sec:indist-ablation}). We report both
regimes; we do not claim the corrector helps on an arbitrary backbone, only that it helps on a
strong one. The trust signal, by contrast, is positive and substantial across all three backbones tested
($\rho=0.611$ SOTA Transolver, $\approx 0.40$--$0.83$ weak grid backbone, $0.851$ bare
MeshGraphNet GNN). We make no competitiveness claim for the grid backbone (it
trails SOTA $\sim 4$--$60\times$); the Transolver backbone we run \emph{on} is itself
competitive.
\item \textbf{Grid resolution and near-wall blindness.} A uniform $128^2$ Cartesian grid
cannot resolve a Re $\approx 10^6$ boundary layer (sub-cell), so wall quantities are
approximate and the absolute MSE values are not comparable to body-fitted solvers.
Relatedly, the monitored residual zeroes the solid-adjacent (wall-ring) cells---where the
central stencil reaches into the $u\!=\!v\!=\!0$ discontinuity and registers spurious
values---so the trust signal operates on the bulk and is structurally blind to the immediate
near-wall band, precisely where near-wall error concentrates on a high-Re airfoil.
\item \textbf{Two 2-D datasets; not yet 3-D.} The trust signal is validated on two 2-D
datasets---turbulent-RANS AirfRANS airfoils and laminar DeepCFD bluff bodies
(\autoref{sec:ood}, $\rho=0.77{\pm}0.12$)---so the phenomenon is not AirfRANS-specific, but all
results remain 2-D. Extending the residual verifier to 3-D (AhmedML, \citealp{ahmedml2024};
DrivAerNet++, \citealp{elrefaie2024drivaernet}) requires 3-D residual operators and multi-TB,
often unsteady datasets beyond a single-GPU budget; we leave it as genuine future work rather
than a near-term claim.
\item \textbf{Conformal coverage out-of-distribution is empirical, not guaranteed.} The
distribution-free guarantee assumes exchangeability, which fails under shift; OOD coverage
holds in our experiments only because the uncertainty inflates appropriately
(\autoref{sec:conformal}). Scores are also pooled over spatially-correlated cells, so the
effective sample size is below the raw cell count, and the calibration set is small
($n=100$); a larger-calibration-set robustness check is future work.
\item \textbf{Residual-as-objective remains a dead end (the gate does not).} Used as a
correction \emph{objective} the residual fails: more iterations raise the PDE residual while
error falls (\autoref{sec:iters}). The DEQ contraction guarantee is in the correction
variable $\delta$, not in the output's PDE residual. The \emph{acceptance gate} built on the
residual is a different mechanism and we no longer claim it is vacuous: measured, it admits
a step on $99.8\%$ of deployed cases and that step lowers true error on $89.3\%$ of them
(\autoref{sec:iters}). That measurement is a single gated step on the deployed corrector,
not the multi-iteration feed-forward loop, which remains untested at this granularity. The correction gains in \autoref{sec:v2} come from
the \emph{learned, supervised} correction, not from minimising the residual---and, by the W1
control, not from feeding the corrector the residual either.
\item \textbf{The corrector's gain is not attributable to the residual input.} The deployed
DEQ corrector ingests the residual as an input channel, but a controlled ablation that zeros
that input at train and inference (5 seeds) matches or slightly exceeds it ($\approx 4\%$ on a
$\approx 0.12$ base; WITH beats NULL on $0/5$ seeds; \texttt{results/control/w1\_capture.json}).
We therefore attribute the $8$--$25\%$ field-MSE improvement to the learned correction
architecture, not to residual-conditioning. The residual's demonstrated value is as a trust
signal (\emph{where} the prediction is wrong), not as a correction input (\emph{how} to fix
it); the deployed system retains the residual input but does not depend on it for the gain.
\item \textbf{Adaptive uncertainty on the SOTA model requires a deep ensemble.} On the
near-deterministic Transolver backbone, MC-dropout $\sigma$ at dropout $0.05$ is
near-degenerate, so although conformal coverage holds ($0.902$, target $0.90$), the
quantile $q$ is enormous ($\approx 10^9$) and ECE $\approx 0.31$: coverage there is
delivered by a near-constant band. A deep ensemble of $M{=}5$ members recovers a genuinely
input-adaptive band on this model ($q=2.35$, ECE $0.074$; \autoref{tab:uq}), at $5\times$
the training cost. That ensemble certificate comes from a single ensemble training run;
the calibration/test split is no longer the caveat (20 re-draws: coverage
$0.899{\pm}0.009$, \autoref{sec:conformal}), and neither is the choice of $M$: sweeping
\emph{every} member subset of size $M=2..5$ (26 subsets;
\texttt{scripts/run\_ensemble\_mstudy.py}, \texttt{results/uq\_ensemble/mstudy.json})
shows the \emph{decision} layer is insensitive to ensemble size---fused AUROC
$0.905$--$0.912$ and conformal coverage $0.905$--$0.915$ (pressure, $0.90$ target) at
every $M$, so even $M{=}2$ supports triage---while band \emph{adaptivity} is what $M$
buys: the pressure quantile inflates $2.35 \rightarrow 3.74 \rightarrow 9.98$ and ECE
degrades $0.074 \rightarrow 0.113 \rightarrow 0.199$ as $M$ drops $5 \rightarrow 3
\rightarrow 2$ (coverage holds by widening, not by adapting), with $M{=}4$ already close
to $M{=}5$ ($q=2.75$, ECE $0.087$). Independently retrained ensembles (beyond member
subsets of one training run) remain future work.
\item \textbf{Eddy-viscosity accuracy is backbone-dependent.} Measured against $\nu_t$'s
own fluid-masked variance on the AirfRANS test split, the channel is well predicted on
the grid backbone ($R^2 = 0.996$) but only moderately on MeshGraphNet ($R^2 = 0.67$--$0.70$;
\texttt{scripts/check\_nut\_learned.py}). Because $\nu_t\approx5.2\,\nu$ on this data and supplies
$\approx84\%$ of $\nu_{\mathrm{eff}}$, a backbone with a weak $\nu_t$ carries that error straight into
the monitored residual---so the audit's own operator is only as turbulent as the closure
the surrogate predicts. Sharpening $\nu_t$ on the weaker backbones is future work.
\item \textbf{Classical fallback is a stub interface.} The trust-gated fallback fixes the
integration seam and reports what would run; the OpenFOAM/SU2 backends are not implemented
and bear on no quantitative claim.
\item \textbf{Sensitivity sweeps are light.} The iteration and toggle sweeps
(\autoref{sec:iters}) are single-checkpoint, unseeded sweeps; they illustrate the
residual--error divergence robustly in direction but are not multi-seed.
\end{itemize}

\section{Conclusion}

We set out to use the steady-RANS physics residual of a neural-CFD surrogate, and we found
that it plays \textbf{two distinct roles with two distinct verdicts}. \emph{As a trust
signal} it is a reliable, backbone-robust (three architecturally-distinct backbones) case-level detector of error---it tells you \emph{where} a prediction is wrong: $\rho=0.611{\pm}0.054$ (five seeds, positive on every seed) on a
SOTA Transolver backbone, $\approx 0.40$--$0.83$ on a weak grid backbone, $0.851{\pm}0.058$ on a
bare message-passing GNN (MeshGraphNet), and---under distribution
shift on the weak backbone---regime-invariant ($\approx 0.74$, rescuing an \texttt{aoa}
collapse from $0.31$ to $0.75$) where the one-shot signal collapses. \emph{As
a correction objective} it fails: more iterations raise the residual while error falls---minimising the residual is
not minimising the error. The monotone-residual \emph{gate} built on that same residual,
by contrast, we measure to be real and useful: it admits a step on $99.8\%$ of deployed
cases and that step lowers true error on $89.3\%$ of them by a median $5.8\%$. Alongside the trust signal we deploy a learned, supervised
correction that works on a strong backbone: the DEQ corrector cuts volume field MSE by
$8$--$25\%$ on all 5 seeds of the SOTA Transolver model (gate-verified), with force ranking
held; a controlled ablation attributes that gain to the learned correction rather than to the
residual input (the residual tells you \emph{where} to look, not \emph{how} to fix). On the strength of the detector result we package a backbone-robust
conformal trust layer that holds target coverage both in-distribution
($0.91/0.93/0.94$ at $0.90$, weak backbone) and on the SOTA backbone ($0.902$), and empirically
retains coverage out-of-distribution via uncertainty inflation; pairing it with a deep-ensemble
$\sigma$ makes the band input-adaptive even on the near-deterministic SOTA model ($q=2.35$,
ECE $0.074$), where MC-dropout alone gives a near-constant one. Two honest caveats
accompany the corrector: correction quality is backbone-dependent (flat on the weak backbone),
and a controlled ablation attributes its gain to the learned correction rather than to the
residual input it ingests. By the criteria of the self-auditing definition in the
introduction, the surrogate passes all three audits it can pass---a validated trust score
(i), a calibrated band that survives split-resampling (ii), and near-oracle triage
(iii)---while the residual-floor theorem fixes the audit's blind spot and the failed
\emph{objective} role marks its licence boundary: the audit tells you when \emph{not} to
trust a prediction, and a separately-supervised corrector, not residual minimisation, does
the fixing---though the cheap monotone-residual gate is a genuine safety net around it. That self-auditing, backbone-robust trust layer and the learned corrector it
accompanies---demonstrated on a competitive surrogate, reproducible end-to-end from the
committed harness---are the durable contribution.

\section*{Code and data availability}
The complete source, training/evaluation harnesses, and per-claim reproduction map
(\texttt{docs/REPRODUCE.md}) are released as the open-source package \texttt{neuroforge-cfd}
(\url{https://github.com/ali-kin4/neuroforge-cfd}; archived at Zenodo,
DOI \href{https://doi.org/10.5281/zenodo.21277928}{10.5281/zenodo.21277928}). Every headline number maps to a committed script and result file, and a top-level
manifest (\texttt{results/MANIFEST.json}) records seeds, environment, and SHA-256 hashes for
hash-drift detection. The two benchmarks are public: AirfRANS \citep{bonnet2022airfrans} and
DeepCFD \citep{ribeiro2020deepcfd}. The package is CPU-first and runs end-to-end with zero
downloads via a synthetic data generator; the GPU runs (backbone training) are documented with
exact commands and expected outputs.

\section*{CRediT author contribution statement}
\textbf{Ali Jabbary:} Conceptualization, Methodology, Software, Validation, Formal analysis,
Investigation, Data curation, Writing -- original draft, Writing -- review \& editing,
Visualization. \textbf{Kasra Ghanavati:} Validation, Writing -- review \& editing.

\section*{Declaration of competing interests}
The authors declare no competing financial or personal interests.

\section*{Declaration of generative AI and AI-assisted technologies in the manuscript
preparation process}
During the preparation of this work the authors used Claude (Anthropic) to assist in
drafting and editing the manuscript text. After using this tool, the authors reviewed and
edited the content as needed and take full responsibility for the content of the
published article. The use of AI assistance in developing the research software is
described in the Implementation section.

\section*{Funding}
This research received no specific grant from funding agencies in the public, commercial,
or not-for-profit sectors. All experiments reported here were run on a single consumer
workstation (one NVIDIA RTX 4070 Ti for backbone training; every audit-side measurement is
CPU-only), which is part of why the reproduction package is CPU-first.

\section*{Acknowledgements}
We thank the authors of AirfRANS \citep{bonnet2022airfrans} and DeepCFD
\citep{ribeiro2020deepcfd} for releasing their datasets, benchmark splits, and reference
force labels publicly; the cross-dataset and force-recovery results in this paper exist
only because those artefacts are open. We also thank the maintainers of the open-source
scientific Python and PyTorch ecosystems on which the released package depends. Assistance
from generative AI tools in preparing the manuscript text is declared separately above.

\bibliographystyle{tmlr}
\bibliography{refs}

\begin{thebibliography}{35}
\providecommand{\natexlab}[1]{#1}
\providecommand{\url}[1]{\texttt{#1}}
\expandafter\ifx\csname urlstyle\endcsname\relax
  \providecommand{\doi}[1]{doi: #1}\else
  \providecommand{\doi}{doi: \begingroup \urlstyle{rm}\Url}\fi

\bibitem[Ashton et~al.(2024)Ashton, Maddix, Gundry, and
  Shabestari]{ahmedml2024}
Neil Ashton, Danielle Maddix, Samuel Gundry, and Parisa Shabestari.
\newblock {AhmedML}: High-fidelity computational fluid dynamics dataset for
  incompressible, low-speed bluff body aerodynamics.
\newblock \emph{arXiv preprint arXiv:2407.20801}, 2024.

\bibitem[Bai et~al.(2019)Bai, Kolter, and Koltun]{bai2019deq}
Shaojie Bai, J.~Zico Kolter, and Vladlen Koltun.
\newblock Deep equilibrium models.
\newblock In \emph{Advances in Neural Information Processing Systems
  (NeurIPS)}, 2019.
\newblock arXiv:1909.01377.

\bibitem[Becker \& Rannacher(2001)Becker and Rannacher]{beckerrannacher2001dwr}
Roland Becker and Rolf Rannacher.
\newblock An optimal control approach to a posteriori error estimation in
  finite element methods.
\newblock \emph{Acta Numerica}, 10:\penalty0 1--102, 2001.
\newblock \doi{10.1017/S0962492901000010}.

\bibitem[Bonnet et~al.(2022)Bonnet, Mazari, Cinnella, and
  Gallinari]{bonnet2022airfrans}
Florent Bonnet, Jocelyn~Ahmed Mazari, Paola Cinnella, and Patrick Gallinari.
\newblock {AirfRANS}: High fidelity computational fluid dynamics dataset for
  approximating reynolds-averaged navier-stokes solutions.
\newblock In \emph{Advances in Neural Information Processing Systems (NeurIPS)
  Datasets and Benchmarks Track}, 2022.
\newblock arXiv:2212.07564.

\bibitem[Elrefaie et~al.(2024)Elrefaie, Morar, Dai, and
  Ahmed]{elrefaie2024drivaernet}
Mohamed Elrefaie, Florin Morar, Angela Dai, and Faez Ahmed.
\newblock {DrivAerNet++}: A large-scale multimodal car dataset with
  computational fluid dynamics simulations and deep learning benchmarks.
\newblock \emph{arXiv preprint arXiv:2406.09624}, 2024.

\bibitem[Fung et~al.(2022)Fung, Heaton, Li, McKenzie, Osher, and
  Yin]{fung2022jfb}
Samy~Wu Fung, Howard Heaton, Qiuwei Li, Daniel McKenzie, Stanley Osher, and
  Wotao Yin.
\newblock {JFB}: Jacobian-free backpropagation for implicit networks.
\newblock In \emph{AAAI Conference on Artificial Intelligence}, 2022.
\newblock arXiv:2103.12803.

\bibitem[Gal \& Ghahramani(2016)Gal and Ghahramani]{gal2016dropout}
Yarin Gal and Zoubin Ghahramani.
\newblock Dropout as a bayesian approximation: Representing model uncertainty
  in deep learning.
\newblock In \emph{International Conference on Machine Learning (ICML)}, 2016.
\newblock arXiv:1506.02142.

\bibitem[Garg \& Chakraborty(2025)Garg and Chakraborty]{garg2025dfuq}
Shailesh Garg and Souvik Chakraborty.
\newblock Distribution free uncertainty quantification for
  neuroscience-inspired deep neural operators.
\newblock \emph{Journal of Computational Physics}, 534:\penalty0 114012, 2025.
\newblock \doi{10.1016/j.jcp.2025.114012}.

\bibitem[Gopakumar et~al.(2025)Gopakumar, Gray, Zanisi, Nunn, Giles, Kusner,
  Pamela, and Deisenroth]{gopakumar2025pre}
Vignesh Gopakumar, Ander Gray, Lorenzo Zanisi, Timothy Nunn, Daniel Giles,
  Matt~J. Kusner, Stanislas Pamela, and Marc~Peter Deisenroth.
\newblock Calibrated physics-informed uncertainty quantification.
\newblock \emph{arXiv preprint arXiv:2502.04406}, 2025.

\bibitem[Hillebrecht \& Unger(2025)Hillebrecht and
  Unger]{hillebrecht2025posteriori}
Birgit Hillebrecht and Benjamin Unger.
\newblock Rigorous a posteriori error bounds for {PDE}-defined {PINNs}.
\newblock \emph{IEEE Transactions on Neural Networks and Learning Systems},
  36\penalty0 (1):\penalty0 1583--1593, 2025.
\newblock \doi{10.1109/TNNLS.2023.3335837}.

\bibitem[Hsieh et~al.(2019)Hsieh, Zhao, Eismann, Mirabella, and
  Ermon]{hsieh2019neuralpde}
Jun-Ting Hsieh, Shengjia Zhao, Stephan Eismann, Lucia Mirabella, and Stefano
  Ermon.
\newblock Learning neural {PDE} solvers with convergence guarantees.
\newblock In \emph{International Conference on Learning Representations
  (ICLR)}, 2019.
\newblock arXiv:1906.01200.

\bibitem[Huang \& Perdikaris(2025)Huang and
  Perdikaris]{huang2025physicscorrect}
Xinquan Huang and Paris Perdikaris.
\newblock {PhysicsCorrect}: A training-free approach for stable neural {PDE}
  simulations.
\newblock \emph{arXiv preprint arXiv:2507.02227}, 2025.
\newblock AAAI 2026.

\bibitem[Jha(2024)]{learned2023residualcorrection}
Prashant~K. Jha.
\newblock Residual-based error corrector operator to enhance accuracy and
  reliability of neural operator surrogates of nonlinear variational
  boundary-value problems.
\newblock \emph{Computer Methods in Applied Mechanics and Engineering},
  419:\penalty0 116595, 2024.
\newblock \doi{10.1016/j.cma.2023.116595}.
\newblock arXiv:2306.12047.

\bibitem[Jia et~al.(2026)Jia, Xia, Vdovin, Jia, Sebben, and
  Yang]{jia2026multigranularity}
Chundong Jia, Chao Xia, Alexey Vdovin, Qing Jia, Simone Sebben, and Zhigang
  Yang.
\newblock Multi-granularity conformal prediction for reliable neural-operator
  automotive aerodynamic surrogates.
\newblock \emph{arXiv preprint arXiv:2607.17297}, 2026.

\bibitem[Krishnapriyan et~al.(2021)Krishnapriyan, Gholami, Zhe, Kirby, and
  Mahoney]{krishnapriyan2021failure}
Aditi~S. Krishnapriyan, Amir Gholami, Shandian Zhe, Robert~M. Kirby, and
  Michael~W. Mahoney.
\newblock Characterizing possible failure modes in physics-informed neural
  networks.
\newblock In \emph{Advances in Neural Information Processing Systems
  (NeurIPS)}, 2021.

\bibitem[Lakshminarayanan et~al.(2017)Lakshminarayanan, Pritzel, and
  Blundell]{lakshminarayanan2017ensembles}
Balaji Lakshminarayanan, Alexander Pritzel, and Charles Blundell.
\newblock Simple and scalable predictive uncertainty estimation using deep
  ensembles.
\newblock In \emph{Advances in Neural Information Processing Systems
  (NeurIPS)}, 2017.
\newblock arXiv:1612.01474.

\bibitem[Li et~al.(2021)Li, Kovachki, Azizzadenesheli, Liu, Bhattacharya,
  Stuart, and Anandkumar]{li2021fno}
Zongyi Li, Nikola Kovachki, Kamyar Azizzadenesheli, Burigede Liu, Kaushik
  Bhattacharya, Andrew Stuart, and Anima Anandkumar.
\newblock Fourier neural operator for parametric partial differential
  equations.
\newblock In \emph{International Conference on Learning Representations
  (ICLR)}, 2021.
\newblock arXiv:2010.08895.

\bibitem[Li et~al.(2022)Li, Huang, Liu, and Anandkumar]{li2022geofno}
Zongyi Li, Daniel~Z. Huang, Burigede Liu, and Anima Anandkumar.
\newblock Fourier neural operator with learned deformations for pdes on general
  geometries.
\newblock \emph{arXiv preprint arXiv:2207.05209}, 2022.

\bibitem[Lippe et~al.(2023)Lippe, Veeling, Perdikaris, Turner, and
  Brandstetter]{lippe2023pderefiner}
Phillip Lippe, Bastiaan~S. Veeling, Paris Perdikaris, Richard~E. Turner, and
  Johannes Brandstetter.
\newblock {PDE}-refiner: Achieving accurate long rollouts with neural {PDE}
  solvers.
\newblock In \emph{Advances in Neural Information Processing Systems
  (NeurIPS)}, 2023.
\newblock arXiv:2308.05732.

\bibitem[Luo et~al.(2025)Luo, Wu, Zhou, Xing, Di, Wang, and
  Long]{luo2025transolverpp}
Huakun Luo, Haixu Wu, Hang Zhou, Lanxiang Xing, Yichen Di, Jianmin Wang, and
  Mingsheng Long.
\newblock Transolver++: An accurate neural solver for {PDE}s on million-scale
  geometries.
\newblock \emph{arXiv preprint arXiv:2502.02414}, 2025.

\bibitem[Ma et~al.(2024)Ma, Azizzadenesheli, and Anandkumar]{ma2024uqno}
Ziqi Ma, Kamyar Azizzadenesheli, and Anima Anandkumar.
\newblock Calibrated uncertainty quantification for operator learning via
  conformal prediction.
\newblock \emph{Transactions on Machine Learning Research (TMLR)}, 2024.
\newblock arXiv:2402.01960.

\bibitem[Marwah et~al.(2023)Marwah, Pokle, Kolter, Lipton, Lu, and
  Risteski]{marwah2023fnodeq}
Tanya Marwah, Ashwini Pokle, J.~Zico Kolter, Zachary~C. Lipton, Jianfeng Lu,
  and Andrej Risteski.
\newblock Deep equilibrium based neural operators for steady-state {PDE}s.
\newblock In \emph{Advances in Neural Information Processing Systems
  (NeurIPS)}, 2023.
\newblock arXiv:2312.00234.

\bibitem[Mukherjee et~al.(2026)Mukherjee, Fitzsimmons, Del Rey~Fern{\'a}ndez,
  and Liu]{mukherjee2026certification}
Amartya Mukherjee, Maxwell Fitzsimmons, David~C. Del Rey~Fern{\'a}ndez, and Jun
  Liu.
\newblock Rigorous error certification for neural {PDE} solvers: From empirical
  residuals to solution guarantees.
\newblock \emph{arXiv preprint arXiv:2603.19165}, 2026.

\bibitem[Raissi et~al.(2019)Raissi, Perdikaris, and
  Karniadakis]{raissi2019pinn}
Maziar Raissi, Paris Perdikaris, and George~E. Karniadakis.
\newblock Physics-informed neural networks: A deep learning framework for
  solving forward and inverse problems involving nonlinear partial differential
  equations.
\newblock \emph{Journal of Computational Physics}, 378:\penalty0 686--707,
  2019.

\bibitem[Ranade et~al.(2025)Ranade, Nabian, Tangsali, Kamenev, Hennigh,
  Cherukuri, and Choudhry]{ranade2025domino}
Rishikesh Ranade, Mohammad~Amin Nabian, Kaustubh Tangsali, Alexey Kamenev,
  Oliver Hennigh, Ram Cherukuri, and Sanjay Choudhry.
\newblock {DoMINO}: A decomposable multi-scale iterative neural operator for
  modeling large scale engineering simulations.
\newblock \emph{arXiv preprint arXiv:2501.13350}, 2025.
\newblock NVIDIA.

\bibitem[Ribeiro et~al.(2020)Ribeiro, Rehman, Ahmed, and
  Dengel]{ribeiro2020deepcfd}
Mateus~Dias Ribeiro, Abdul Rehman, Sheraz Ahmed, and Andreas Dengel.
\newblock {DeepCFD}: Efficient steady-state laminar flow approximation with
  deep convolutional neural networks.
\newblock \emph{arXiv preprint arXiv:2004.08826}, 2020.

\bibitem[Rigotti et~al.(2026)Rigotti, Thumiger, and Frick]{rigotti2026gist}
Mattia Rigotti, Nicholas Thumiger, and Thomas Frick.
\newblock {GIST}: Gauge-invariant spectral transformers for scalable graph
  neural operators.
\newblock \emph{arXiv preprint arXiv:2603.16849}, 2026.

\bibitem[Roy et~al.(2025)Roy, Nayak, and Goswami]{roy2025anchor}
Rajyasri Roy, Dibyajyoti Nayak, and Somdatta Goswami.
\newblock {ANCHOR}: Error-controlled adaptive numerical correction for neural
  operator time marching.
\newblock \emph{arXiv preprint arXiv:2512.19643}, 2025.

\bibitem[Scherz et~al.(2026)Scherz, Hines, and Bekemeyer]{scherz2026evaluation}
Jan Scherz, Derrick Hines, and Philipp Bekemeyer.
\newblock Evaluation of state-of-the-art deep learning architectures for
  aerodynamical predictions.
\newblock \emph{arXiv preprint arXiv:2607.13866}, 2026.

\bibitem[Song et~al.(2026)Song, Li, Deng, Li, Pan, Lai, and
  Wang]{song2026structureaware}
Haoze Song, Zhihao Li, Mengyi Deng, Xin Li, Duyi Pan, Zhilu Lai, and Wei Wang.
\newblock Structure-aware epistemic uncertainty quantification for neural
  operator {PDE} surrogates.
\newblock \emph{arXiv preprint arXiv:2603.11052}, 2026.

\bibitem[Trefethen \& Bau(1997)Trefethen and Bau]{trefethenbau1997}
Lloyd~N. Trefethen and David Bau.
\newblock \emph{Numerical Linear Algebra}.
\newblock SIAM, 1997.

\bibitem[Wang et~al.(2022)Wang, Yu, and Perdikaris]{wang2022ntk}
Sifan Wang, Xinling Yu, and Paris Perdikaris.
\newblock When and why {PINNs} fail to train: A neural tangent kernel
  perspective.
\newblock \emph{Journal of Computational Physics}, 449:\penalty0 110768, 2022.

\bibitem[Winston \& Kolter(2020)Winston and Kolter]{winston2020monotone}
Ezra Winston and J.~Zico Kolter.
\newblock Monotone operator equilibrium networks.
\newblock In \emph{Advances in Neural Information Processing Systems
  (NeurIPS)}, 2020.
\newblock arXiv:2006.08591.

\bibitem[Wu et~al.(2024)Wu, Luo, Wang, Wang, and Long]{wu2024transolver}
Haixu Wu, Huakun Luo, Haowen Wang, Jianmin Wang, and Mingsheng Long.
\newblock Transolver: A fast transformer solver for {PDE}s on general
  geometries.
\newblock In \emph{International Conference on Machine Learning (ICML)}, 2024.
\newblock arXiv:2402.02366.

\bibitem[Yu et~al.(2026)Yu, Ho, and Wang]{yu2026conformalpinn}
Yifan Yu, Cheuk~Hin Ho, and Yangshuai Wang.
\newblock A conformal prediction framework for uncertainty quantification in
  physics-informed neural networks.
\newblock \emph{Journal of Computational Physics}, 561:\penalty0 114979, 2026.
\newblock \doi{10.1016/j.jcp.2026.114979}.

\end{thebibliography}

\end{document}